\documentclass[runningheads]{llncs}
\usepackage{fullpage}
\usepackage[T1]{fontenc}
\usepackage{graphicx}
\usepackage{amsmath}

\usepackage{thm-restate}
\usepackage{amsthm}
\usepackage{amssymb,amsfonts,latexsym}

\usepackage{tabularx,environ}
\usepackage{enumerate}
\usepackage{bm}
\usepackage{algorithm}
\usepackage{algpseudocode}
\usepackage{paralist}

\usepackage{tikz}
\usetikzlibrary{calc,positioning,shapes}
\usepackage[hidelinks]{hyperref}
\usepackage[mathlines]{lineno}
\usepackage[capitalise]{cleveref}
\usepackage{empheq}

\usepackage{multibib}
\newcites{app}{Additional references in the appendix}

\theoremstyle{remark}

\newtheorem{clm}{Claim}
\newtheorem{obs}{Observation}
\usepackage{comment}
\usepackage{makecell}

\newcommand{\graph}{G=(V,E)}

\newcommand{\mcst}{MCST}
\newcommand{\spmsat}{SPM $3$-Bounded $3$-SAT}
\newcommand{\recrep}{\mathcal{G}(\varphi)}

\newcommand{\domnum}[1]{\gamma{(#1)}}
\newcommand{\vcnum}[1]{\tau_{ST}{(#1)}}

\newcommand{\pow}{\mathcal{P}}
\newcommand{\decom}{\mathcal{T}}
\newcommand{\distg}{\text{dist}_{G}}
\newcommand{\cw}{\mathsf{cw}}
\newcommand{\Input}[1]{\State \textbf{Input:} #1}
\newcommand{\Output}[1]{\State \textbf{Output:} #1}
\newcommand{\hojoG}{D}
\newcommand{\nanka}{P}

\newcommand{\adj}[2]{\mathsf{adj}(#1,#2)}

\newcommand{\Vint}[2]{V_{[#1,#2]}}

\begin{document}
\title{Spanning Trees with a Small Vertex Cover: \\the Complexity on Specific Graph Classes\thanks{This work was partially supported by JSPS KAKENHI Grant Numbers {JP25K14980} and JP25K21148.}}
\titlerunning{Spanning Trees with a Small Vertex Cover}
%
\author{
	Toranosuke Kokai\inst{1} \and
	Akira Suzuki\inst{2}\orcidID{0000-0002-5212-0202} \and
	Takahiro Suzuki\inst{1}\orcidID{0009-0005-8433-3789} \and
	Yuma Tamura\inst{1}\orcidID{0009-0001-5479-7006} \and
	Xiao Zhou\inst{1}
}
\authorrunning{T. Kokai et al.}
%
\institute{Graduate School of Information Sciences, Tohoku University, Sendai, Japan \\ \email{\{toranosuke.kokai.s6, takahiro.suzuki.q4\}@dc.tohoku.ac.jp, \\\{tamura, zhou\}@tohoku.ac.jp}\and
Center for Data-driven Science and Artificial Intelligence, Tohoku University, Sendai, Japan \\
\email{akira@tohoku.ac.jp}
}
\maketitle              
\begin{abstract}
In the context of algorithm theory, various studies have been conducted on spanning trees with desirable properties. In this paper, we consider the \textsc{Minimum Cover Spanning Tree} problem (MCST for short). Given a graph $G$ and a positive integer $k$, the problem determines whether $G$ has a spanning tree with a vertex cover of size at most $k$. We reveal the equivalence between \mcst\ and the \textsc{Dominating Set} problem when $G$ is of diameter at most~$2$ or $P_5$-free. This provides the intractability for these graphs and the tractability for several subclasses of $P_5$-free graphs. We also show that \mcst\ is NP-complete for bipartite planar graphs of maximum degree~$4$ and unit disk graphs.
These hardness results resolve open questions posed in prior research.
Finally, we present an FPT algorithm for {\mcst} parameterized by clique-width and a linear-time algorithm for interval graphs.

\end{abstract}
\section{Introduction}
A \emph{spanning tree} of a graph is a fundamental concept in graph theory.
For a connected graph $G$, a spanning tree $T$ is an acyclic connected subgraph of $G$ that contains all vertices in $G$.
It is well known that a spanning tree of $G$ can be found in polynomial time.
Spanning trees play a crucial role in both theoretical science and practical applications.
Theoretically, spanning trees serve as a powerful tool for developing algorithms and are deeply connected to more advanced mathematical theories such as matroids.
In terms of applications, spanning trees are essential for achieving efficient communication in networks.

Spanning trees with specific restrictions have also attracted attention, such as those with few leaves~\cite{SalamonW08}, many leaves~\cite{Rosamond16}, bounded diameter~\cite{HoLCW91}, or small poise (the maximum degree plus the diameter)~\cite{Ravi94}. 
In particular, a spanning tree with two leaves is also known as a Hamiltonian path, a topic of long-standing research in graph theory.
The maximum number of leaves among spanning trees of $G$ is called the \emph{max-leaf number}.
The max-leaf number is closely related to the connected domination number (for example, see~\cite{CaroWY00}).
Moreover, the max-leaf number is also known as a structural parameter of graphs, and efficient FPT algorithms have been developed for various problems~\cite{FellowsLMMRS09}.

In this paper, we deal with spanning trees with a small vertex cover.
A vertex subset $S$ of a graph $G$ is a \emph{vertex cover} if at least one endpoint of every edge of $G$ belongs to $S$.
Given a positive integer $k$, the \textsc{Minimum Cover Spanning Tree} problem (MCST for short) asks for a spanning tree $T$ with a vertex cover of size at most $k$.
Intuitively, spanning trees with a small vertex cover number tend to have a star-like structure and a small diameter. 
Spanning trees of such type are considered easy to handle in practical applications.
For example, a complete graph $G$ with $n$ vertices admits a spanning star $T$ with a vertex cover of size~$1$, whereas it also admits a spanning path $T'$ with a vertex cover of size~$\lfloor n/2 \rfloor$ (see \cref{fig:complete}).
If the vertices and edges of $G$ represent the communication nodes and links of a network, then communication along $T$ is expected to have lower latency and to be easier to monitor.
\begin{figure}[t]
\centering
\includegraphics[width=0.7\linewidth]{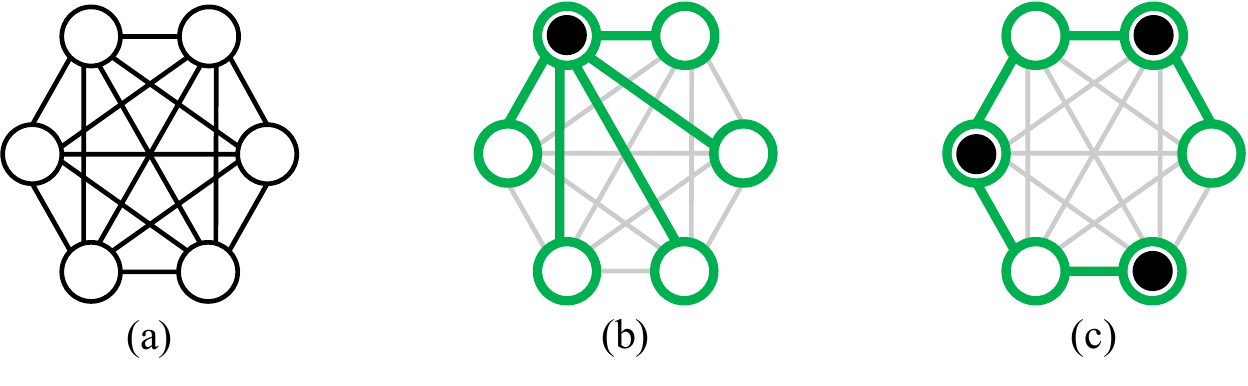}
\caption{(a) A Complete graph $G$, (b) a spanning tree of $G$ with a vertex cover of size~$1$, and (c) a spanning tree of $G$ with a vertex cover of size~$3$.
The spanning trees are shown in bold green lines, and the corresponding vertex covers are indicated by black circles.}
\label{fig:complete}
\end{figure}

\subsection{Known results}\label{subsec:knownresults}
Angel et al.\ showed that any graph $G$ satisfies $\domnum{G} \leq \vcnum{G} \leq 2\domnum{G} -1$~\cite{AngelBCK15}, where $\vcnum{G}$ denotes the minimum size of vertex covers among spanning trees of $G$ and $\domnum{G}$ denotes the minimum size of dominating sets in $G$.
From the (in)approximability of the \textsc{Dominating Set} problem,
it follows that a polynomial-time $O(\log n)$-approximation algorithm exists for \mcst, whereas there exists a constant $c$ such that \mcst\ is NP-hard to approximate within a factor of $c\log n$.

Fukunaga and Maehara gave an alternative proof of the above inapproximability of \mcst~\cite{FukunagaM19}.
They also provided a polynomial-time constant-factor approximation algorithm for \mcst\ on unit disk graphs and graphs excluding a fixed minor.

Kaur and Misra investigated \mcst\ from the perspective of parameterized complexity~\cite{KaurM20}.
\mcst\ was shown to be W[2]-hard when parameterized by the vertex cover size $k$ in spanning trees, even for bipartite graphs.
On the other hand, an FPT algorithm parameterized by treewidth was provided.
They also studied a kernelization and showed that the problem parameterized by $k$ does not admit a polynomial kernel unless $\mathrm{coNP} \subseteq \mathrm{NP/ poly}$.

\subsection{Our contribution}
\begin{figure}[t]
    \centering
    \includegraphics[width=0.9\linewidth]{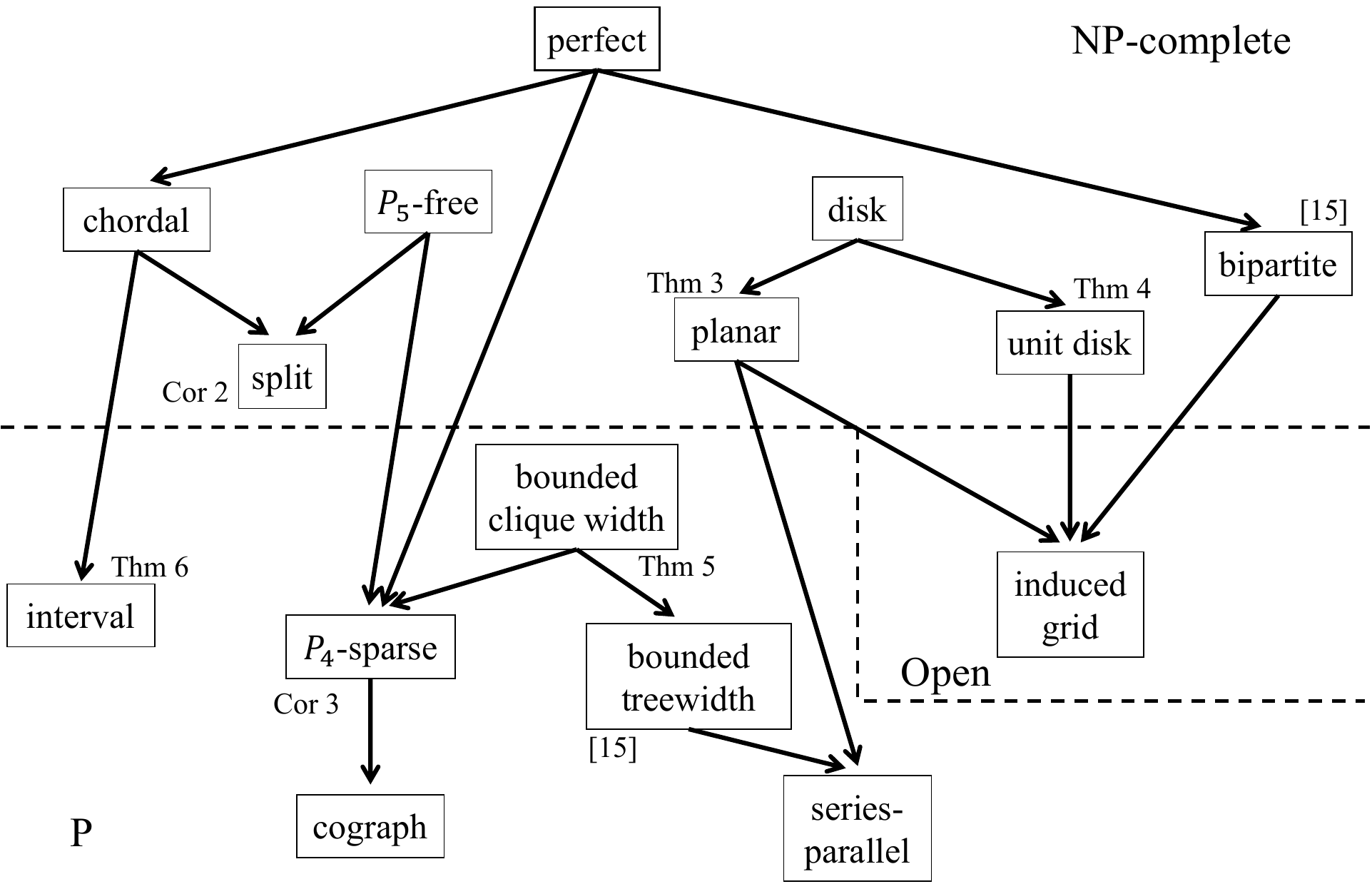}
    \caption{Known and our results with respect to graph classes. Each arrow $A \rightarrow B$ represents that the graph class $B$ is a subclass of the graph class $A$.}
    \label{fig:graphclass}
\end{figure}
This paper provides a detailed study of \mcst\ from the viewpoint of graph classes. (See~\cref{fig:graphclass}.)
Denote $P_n$ by the path with $n$ vertices.
We first investigate the relationship between $\vcnum{G}$ and $\domnum{G}$.
We reveal that $\vcnum{G} = \domnum{G}$ holds when a graph $G$ has diameter at most~$2$ or is $P_5$-free.
These findings have various algorithmic consequences.
We show that \mcst\ on split graphs with $n$ vertices cannot be approximated to within a factor of $(1-\varepsilon)\log n$ in polynomial time for any constant $0 < \varepsilon < 1$ unless $\mathrm{NP} \subseteq \mathrm{DTIME}(n^{O(\log \log n)})$.
Moreover, it can be shown that \mcst\ remains W[2]-hard even for split graphs of diameter~$2$ when parameterized by $\vcnum{G}$.
On the other hand, if \textsc{Dominating Set} is solvable in polynomial time for a subclass of 
$P_5$-free graphs, then \mcst\ can also be solved in polynomial time for the same graph class.
Such graph classes include cographs, $P_4$-sparse graphs, chain graphs, co-bipartite graphs, and co-interval graphs. 
We emphasize that our result is tight in the sense that there is a $P_6$-free graph $G$ of diameter~$3$ such that $\vcnum{G} \neq \domnum{G}$.

The above result motivates us to analyze the computational complexity of \mcst\ on graph classes that admit long paths: planar graphs, bipartite graphs, unit disk graphs, bounded clique-width graphs, and interval graphs.
We show that \mcst\ is NP-complete for planar bipartite graphs with maximum degree at most~$4$ and unit disk graphs.
These results resolve the open question posed by Fukunaga and Maehara~\cite{FukunagaM19}.
On the other hand, we reveal that \mcst\ is fixed-parameter tractable when parameterized by clique-width, which extends the known result that \mcst\ is fixed-parameter tractable for treewidth~\cite{KaurM20}.
This can be proved by defining a problem equivalent to \mcst\ that is expressible in monadic second-order logic (MSO$_1$).
A result of Courcelle et al.\ implies that the equivalent problem is fixed-parameter tractable for clique-width~\cite{Courcelletheorem}. 
We also design an FPT algorithm parameterized by clique-width to improve the running time.
Finally, we present a linear-time algorithm for interval graphs.

The proofs marked $(\ast)$ are provided in Appendix.
\paragraph{Future work.}
One interesting open question is whether {\mcst} is solvable in polynomial time for induced grid graphs (i.e.\ induced subgraphs of the Cartesian product of two paths).
We note that \textsc{Dominating Set} is NP-complete for induced grid graphs, whereas \textsc{Vertex cover} can be solved in polynomial time for bipartite graphs.

\section{Preliminaries}
For a positive integer $n$, we write $[n]=\{1,\dots,n\}$.
Let {$\graph$} be a graph: we also denote by $V(G)$ and $E(G)$ the vertex set and the edge set of $G$, respectively. All the graphs considered in this paper are finite, simple, and undirected. 
For two vertices $u,v$ in a graph $G$, denote by $\distg(u,v)$ the length of a shortest path between $u$ and $v$ in $G$.
The \emph{diameter} of $G$ is defined as $\max_{u,v\in V(G)}\distg(u,v)$.
A graph is said to be \emph{connected} if there exists a path between every pair of vertices $u,v \in V(G)$.
A \emph{component} of $G$ is a maximal connected subgraph of $G$.
For a vertex $v$ of $G$, we denote by $N_G(v)$ and $N_G[v]$ the open and closed neighborhood of $v$ in $G$, respectively, that is, $N_G(v)=\{u \in V \mid uv \in E\}$ and $N_G[v]=N_G(v) \cup \{v\}$. 
We sometimes drop the subscript $G$ if it is clear from the context.
The \emph{degree} $\text{deg}_G(v)$ of $v$ is the size of $N_G(v)$, that is, $\text{deg}_G(v)=|N_G(v)|$.
The \emph{maximum degree} of $G$ is defined as $\max_{v\in V} \text{deg}_G(v)$.
A vertex subset $I\subseteq V$ is called an \emph{independent set} if $uv \notin E$ for any two vertices $u,v\in I$.
A \emph{dominating set} of $G$ is a subset $D \subseteq V$ such that every $v \in V\setminus D$ is adjacent to some vertex in $D$.
A graph is \emph{bipartite} if there exists a partition $(I_1, I_2)$ of $V$ such that both $I_1$ and $I_2$ are independent sets.
A graph is \emph{planar} if it can be embedded in the plane without crossing edges.
The graph $G$ is \emph{acyclic}, or equivalently a \emph{forest}, if $G$ contains no cycles.
A \emph{spanning forest} $F$ of $G$ is an acyclic subgraph with $V(F)=V(G)$.
A \emph{spanning tree} $T$ is a connected spanning forest of $G$.
A vertex subset $C\subseteq V(G)$ is called a \emph{vertex cover} if every edge $uv \in E(G)$ satisfies $u\in C$ or $v \in C$.
We also say that the graph $G$ is \emph{covered by} $C$. 

Given a connected graph $G=(V,E)$ and a positive integer $k$, the \textsc{Minimum Cover Spanning Tree (MCST)} problem determines whether $G$ admits a spanning tree $T$ with a vertex cover of size at most $k$.
Here, we introduce a useful notion that will be used throughout the paper.
For a graph $G = (V, E)$ and a subset $S \subseteq V$, the \emph{boundary subgraph} $H=(V,E_H)$
for $S$ is defined as a subgraph of $G$ such that is the set of edges incident to some vertex in $S$.
Note that $S$ is a vertex cover of every subgraph of $H$.
\begin{obs} \label{obs:boundary}
    Let $S$ be a vertex subset of a graph $G$. 
    There exists a spanning tree of $G$ covered by $S$ if and only if the boundary subgraph for $S$ is connected.
\end{obs}

\section{Relationship to the domination number}\label{sec:dom}
The \emph{domination number} of $G$, denoted by $\domnum{G}$, is the size of a minimum dominating set in $G$.
Denote by $\vcnum{G}$ the minimum size of a vertex cover among spanning trees in $G$.
In this section, we show the following theorems.
\begin{theorem}\label{the:bound_dia2}
For any graph $G$ with diameter at most~2, $\domnum{G} = \vcnum{G}$.
\end{theorem}
\begin{proof}
Since $\domnum{G} \le \vcnum{G}$ holds for general graphs~\cite{AngelBCK15}, the task here is to show that $\vcnum{G} \le \domnum{G}$.
Let $S$ be a minimum dominating set of $G = (V, E)$, where $|S| = \domnum{G}$.
Let $H$ be the boundary subgraph for $S$.
From observation~\ref{obs:boundary}, it suffices to prove that $H$ is connected, which leads to $\vcnum{G} \le |S| = \domnum{G}$.

For the sake of contradiction, assume that $H$ has two connected components $H_1$ and $H_2$.
Let $S_1 = S \cap V(H_1)$ and $S_2 = S \cap V(H_2)$.
Note that $S_1 \neq \emptyset$ and $S_2 \neq \emptyset$ from the construction of $H$.
Consider a vertex $v_1 \in S_1$ and a vertex $v_2 \in S_2$.
Since $G$ has diameter at most $2$, the vertices $v_1$ and $v_2$ are adjacent or there exists a vertex $u$ in $G$ such that $\langle v_1, u, v_2 \rangle$ forms a path of $G$.
In both cases, from the definition of $H$, the vertices $v_1$ and $v_2$ are contained in the same connected component in $H$, a contradiction.
\end{proof}

For a graph $H$, a graph $G$ is said to be \emph{$H$-free} if $G$ contains no induced subgraph isomorphic to $H$.
\begin{theorem} \label{the:bound_P5free}
For any connected $P_5$-free graph $G$, $\domnum{G} = \vcnum{G}$.
\end{theorem}
\begin{proof}
Let $H$ be the boundary subgraph for a minimum dominating set $S$ of $G$, where $|S| = \domnum{G}$.
Suppose that $H$ has $c \ge 2$ connected components.
Let $H_1, \dots, H_c$ be connected components of $H$ and denote $S_i = S \cap V(H_i)$ and $R_i = V(H_i) \setminus S$ for $i \in [c]$.
If $S_i = \emptyset$, then the vertices in $R_i$ are not dominated by the vertices in $S$, contradicting that $S$ is a dominating set of $G$.
If $R_i = \emptyset$, then it contradicts the minimality of $S$.
Thus, we have $S_i \neq \emptyset$ and $R_i \neq \emptyset$.

Consider $H_1$.
Then, there exists another connected component, say $H_2$, such that there exist two vertices $r_1 \in R_1$ and $r_2 \in R_2$ with $r_1r_2 \in E(H)$; otherwise, the graph $G$ would not be connected.
Moreover, from the construction of $H$, there exist two vertices $s_1 \in S_1$ and $s_2 \in S_2$ such that $s_1r_1 \in E(H)$ and $r_2s_2 \in E(H)$.
In other words, $P = \langle s_1, r_1, r_2, s_2 \rangle$ forms a path of $G$.
Observe that $s_1r_2, s_1s_2, r_1s_2 \notin E(G)$, since $H_1$ and $H_2$ are distinct connected components.
Consequently, $P$ is in fact an induced path of $G$.

Let $S' = S \cup \{r_1, r_2\} \setminus \{ s_1, s_2\}$, and let $H'$ be the boundary subgraph for $S'$. 
We show the following two claims: (1) $S'$ is a minimum dominating set of $G$; and (2) $H'$ has fewer connected components than the graph $H$.
Iterating the above arguments, eventually the boundary subgraph $H^\ast$ for a minimum dominating set $S^\ast$ of $G$ can be obtained.
It follows from Observation~\ref{obs:boundary} that $G$ has a spanning tree $T$ covered by $S^\ast$.
In conclusion, we have $\vcnum{G} \le |S^\ast| = \domnum{G}$, and \Cref{the:bound_P5free} follows from the discussion in \Cref{subsec:knownresults}.

We give the proof of claim (1).
For the sake of contradiction, assume that there exists a vertex $w \notin S'$ that is not adjacent to any vertex in $S'$.
Since $S$ is a dominating set of $G$, the vertex $w$ is adjacent to $s_1$ or $s_2$.
Without loss of generality, suppose that $w$ is adjacent to $s_1$.
From the assumption of $S'$, $w$ is adjacent to neither $r_1$ nor $r_2$.
Furthermore, $w$ is not adjacent to $s_2$; otherwise, $s_1$ and $s_2$ belong to the same connected component of $H$, a contradiction.
Therefore, $\langle w, s_1, r_1, r_2, s_2 \rangle$ forms an induced path of $G$. 
However, this contradicts that $G$ is $P_5$-free.
This completes the proof of claim (1).

We now prove the claim (2).
First, we show that any distinct vertices in $V(H_1)\cup V(H_2)$ are in the same connected component of $H'$. 
To this end, it suffices to claim that for each vertex $v\in V(H_1)$ (resp.\ $v \in V(H_2)$), there is a $(v,r_1)$-path (resp.\ a $(v,r_2)$-path) in $H'$. 
This is because $r_1r_2\in E(H')$ and the existence of a path is transitive; that is, if there is an $(x,y)$-path and a $(y,z)$-path, then an $(x,z)$-path exists.

Suppose the case $v \in V(H_1)$.
If $v$ is adjacent to $r_1$ in $H$, the claim is trivial because $v$ is also adjacent to $r_1$ in $H'$.
Suppose otherwise.
Since $H_1$ is connected, there is a shortest $(v, s_1)$-path $P'=\langle v, p_1,\ldots, p_\ell, s_1\rangle$ of $H_1$, where $\ell$ is a positive integer and $p_i \in V(H_1)$ for $i \in [\ell]$.
Note that $p_i \neq s_1$ for $i \in [\ell]$.
It follows from $(S\setminus S')\cap V(H_1) = \{s_1\}$ that we have $vp_1 \in E(H')$ and $p_{i}p_{i+1} \in E(H')$ for $i\in [\ell-1]$.
However, it is possible that $p_\ell s_1 \notin E(H')$.
To obtain a $(v,r_1)$-path in $H'$, let us turn our attention to the fact that $p_\ell s_2\notin E(G)$ from the definition of $H_1$ and $H_2$.
This implies that the vertex $p_\ell$ is adjacent to $r_1$ or $r_2$ in $G$ (and more specifically in $H'$ due to $S'$); otherwise, $\langle p_\ell ,s_1,r_1,r_2,s_2\rangle$ forms an induced $P_5$ of $G$, since $P = \langle s_1, r_1, r_2, s_2 \rangle$ is an induced path of $G$.
Thus, $\langle v, p_1, \ldots, p_\ell, r_1\rangle$ or $\langle v, p_1, \ldots, p_\ell, r_2, r_1\rangle$ forms a $(v, r_1)$-path of $H'$.
The same can be applied for $V(H_2)$ and $r_2$, and we can see that there is a $(v,r_2)$-path for every $v\in V(H_2)$.

Furthermore, it should be noted that every vertex $v \in V(H_i) \cap S$ also belongs to $V(H_i) \cap S'$ for $i \in [c]\setminus \{1,2\}$.
This implies that $H_i$ is also a subgraph of $H'$.
Thus, any two vertices in $V(H_i)$ for $i \in [c]\setminus \{1,2\}$ are contained in the same connected component of $H'$.
(It is possible that every vertex in $V(H_i)$ is in the same connected components as the vertices $V(H_1) \cup V(H_2)$.)
Therefore, $H'$ has at most $c-1$ connected components, as claimed.
\end{proof}
\begin{remark}
    Denote by $C_6$ the cycle of length~$6$.
    Observe that $\domnum{C_6} = 2$ and $\vcnum{C_6} = 3$, and hence $\domnum{C_6} \neq \vcnum{C_6}$.
    Furthermore, $C_6$ has diameter~$3$ and is $P_6$-free. 
    In that sense, \Cref{the:bound_dia2,the:bound_P5free} are tight.    
\end{remark}

\subsubsection{Algorithmic consequences.}
The \textsc{Dominating Set} problem asks for a minimum dominating set $ D $ of a graph $ G $.
It is known that \textsc{Dominating Set} on split graphs with $n$ vertices cannot be approximated to within a factor of $(1-\varepsilon)\log n$ in polynomial time for any constant $0 <\varepsilon < 1$ unless $\mathrm{NP} \subseteq \mathrm{DTIME}(n^{O(\log \log n)})$~\cite{ChlebikC08}.
Moreover, when parameterized by $\vcnum{G}$, \textsc{Dominating Set} remains W[2]-hard even for split graphs of diameter~$2$~\cite{LokshtanovMPRS13}, which implies that FPT algorithms are unlikely to exist for these graphs.
Note that split graphs are $P_5$-free.
Consequently, we obtain the following two corollaries from \Cref{the:bound_dia2,the:bound_P5free}.

\begin{corollary}
    \mcst\ on split graphs with $n$ vertices cannot be approximated to within a factor of $(1-\varepsilon)\log n$ in polynomial time for any constant $0<\varepsilon < 1 $ unless $\mathrm{NP} \subseteq \mathrm{DTIME}(n^{O(\log \log n)})$.
\end{corollary}

\begin{corollary}
    \mcst\ remains W[2]-hard even for split graphs of diameter~$2$ when parameterized by $\vcnum{G}$.
\end{corollary}

On the other hand, \textsc{Dominating Set} is solvable in polynomial time for cographs, $P_4$-sparse graphs, chain graphs~\cite{Courcelletheorem},\footnote{This follows from the fact that cographs, $P_4$-sparse graphs, and chain graphs have bounded clique-width~\cite{Courcelletheorem,GolumbicR00}.} co-bipartite graphs~\cite{KratschS93} and co-interval graphs~\cite{KeilB04}.
These graph classes are subclasses of $P_5$-free graphs.
The following corollary is immediately obtained from \Cref{the:bound_P5free}.

\begin{corollary}
    \mcst\ is solvable in polynomial time for cographs, $P_4$-sparse graphs, chain graphs, co-bipartite graphs, and co-interval graphs.
\end{corollary}
\section{NP-completeness}\label{sec:hard}
In this section, we prove that {\mcst} is NP-complete for several graph classes.
We begin by describing the source problem for our reduction.
\subsection{\textsc{Simple Planar Monotone (SPM) 3-bounded 3-SAT}}
Let $\varphi$ be a conjunctive normal form (CNF) formula over Boolean variables $\mathcal{X} = \{x_1, \dots, x_n\}$, given by a conjunction of clauses $\mathcal{C} = \{C_1, \dots, C_m\}$, that is, $\varphi = C_1 \land \cdots \land C_m$.
A CNF formula $\varphi$ is \emph{satisfiable} if there exists a truth assignment to $\mathcal{X}$ that makes $\varphi$ evaluate to true.
A 3-CNF formula is a CNF formula in which each clause $C_j \in \mathcal{C}$ is the disjunction of at most three literals.
Given a 3-CNF formula $\varphi$, the 3-SAT problem decides whether $\varphi$ is satisfiable.
A clause containing only positive literals is called a \emph{positive clause}, and one containing only negative literals is called a \emph{negative clause}.
For a CNF formula $\varphi$, the \emph{incidence graph} $\recrep = (\mathcal{V}, \mathcal{E})$ is a bipartite graph with 
$\mathcal{V}=\mathcal{X} \cup \mathcal{C}$ and $\mathcal{E} = \{(x_i,C_j) \mid x_i \in C_j\text{ or }\overline{x_i}\in C_j\}$.
A \emph{rectilinear representation} of $\recrep$ is a planar drawing in which each variable and clause is represented by a rectangle, all variable vertices lie along a horizontal line, and each edge connecting a variable to a clause is drawn as a vertical line segment with no edge crossings. (See \cref{fig:reduction}(a).)
\textsc{Simple Planar Monotone (SPM) 3-bounded 3-SAT} is a variant of 3-SAT with the following additional restrictions: 
\begin{itemize}
    \item Each clause in $\mathcal{C}$ contains at most one occurrence of each variable (simple);
    \item The incidence graph {$\recrep$} of $\varphi$ has a rectilinear representation (planar); 
    \item Each clause is either a positive or a negative clause (monotone); and
    \item Each variable occurs in at most three clauses (3-bounded).
\end{itemize}

Given a simple planar monotone 3-bounded 3-CNF formula $\varphi$, the \textsc{Simple Planar Monotone 3-bounded 3-SAT} problem determines whether $\varphi$ is satisfiable.
This problem is known to be NP-complete~\cite{doi:10.1142/S0129054118500168}.

When $\varphi$ is planar and monotone, it may be supposed that a rectilinear representation of $\recrep$ is drawn with all positive clauses placed above the variables and all negative clauses placed below them~\cite{DBLP:conf/cocoon/BergK10}.
(See \cref{fig:reduction}(a) again.)
Note that any variable occurring only in positive clauses can be assigned to True.
Similarly, any variable occurring only in negative clauses can be assigned to False.
Thus, we may assume that each variable appears in exactly one or two positive clauses and in exactly one or two negative clauses.

\begin{figure}[t]
\centering
\includegraphics[width=\textwidth]{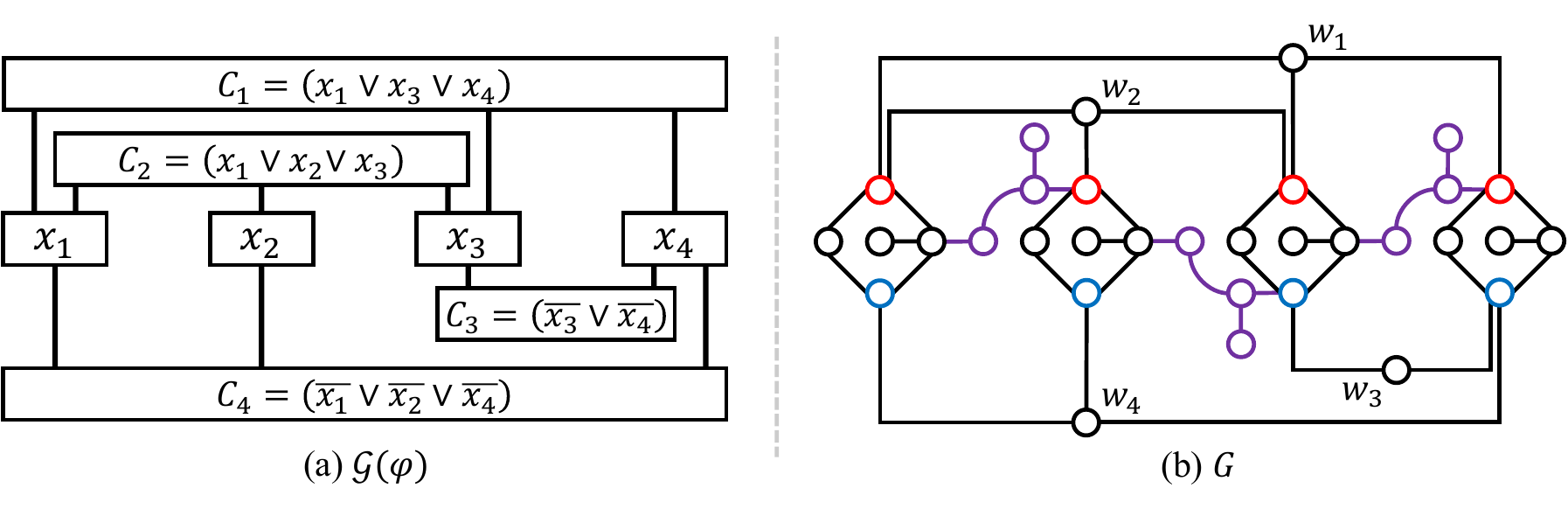}
\caption{(a) An example of illustrating a rectilinear representation of a formula $\varphi = (x_1 \lor x_3 \lor x_4) \land (x_1 \lor x_2 \lor x_3) \lor (\overline{x_3} \lor \overline{x_4}) \land (\overline{x_1} \lor \overline{x_2} \lor \overline{x_4})$ and (b) the graph $G$ constructed from $\varphi$.}
\label{fig:reduction}
\end{figure}
\subsection{Planar bipartite graphs with maximum degree $4$}
We now turn to prove the following theorem.
\begin{theorem}
    {\mcst} is NP-complete even for planar bipartite graphs with maximum degree at most~$4$.
    \label{the:planarhardness}
\end{theorem}
Clearly, {\mcst} belongs to the class NP.
To prove \cref{the:planarhardness}, we provide a polynomial-time reduction from {\spmsat} to {\mcst}. 
Let {$\varphi$} be an instance of {\spmsat} and let $\recrep$ be its rectilinear representation.
We first construct an instance $(G,k)$ of {\mcst} from $\recrep$ with the following gadgets.
\subsubsection{Variable Gadget.}
The variable gadget {$U_i$} for a Boolean variable $x_i\in \mathcal{X}$ consists of the cycle $\langle u_i,r_i,\overline{u_i},t_i\rangle$ and the vertex $s_i$ adjacent to $t_i$, as illustrated in Fig.~\ref{fig:valgad}(a).
Assigning the value True (resp.\ False) to the variable $x_i$ corresponds to including the vertex $u_i$ (resp.\ $\overline{u_i}$) in a vertex cover of a spanning tree of $U_i$.
\begin{figure}[t]
\centering
\includegraphics[width=\linewidth]{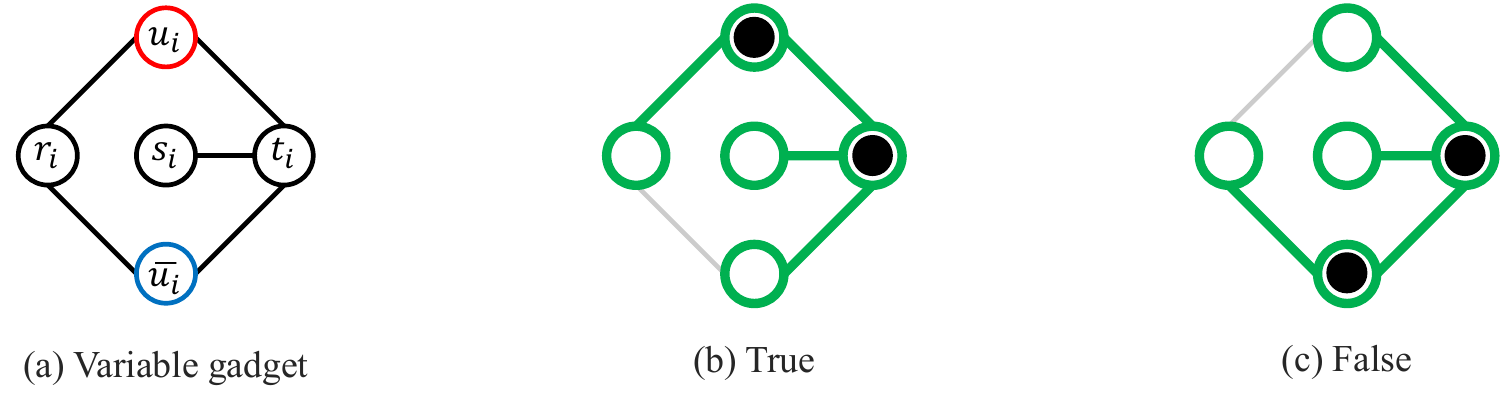}
\caption{(a) The variable gadget $U_i$. If the variable $x_i$ is assigned True (resp.\ False), then the corresponding tree is depicted in green and its vertex cover is indicated by black circles in (b) (resp.\ (c)).} \label{fig:valgad}
\end{figure}
\subsubsection{Connector Gadget.}
The connector gadget is used to connect spanning trees in $U_i$ and $U_{i+1}$.
The connector gadget $H_i$ with $i \in [n-1]$ consists of the path $\langle h_i, q_i,\ell_i\rangle$ together with the edge $h_i t_i$ and the edge either $q_i u_{i+1}$ or $q_i \overline{u_{i+1}}$, as shown in \cref{fig:congad}.
The choice of which edge to include depends on the occurrences of the variable $x_{i+1}$. 
Recall that each variable appears in exactly one or two positive clauses and in exactly one or two negative clauses.
If $x_{i+1}$ occurs exactly once in a positive clause, then we connect $q_i$ to $u_{i+1}$.
Conversely, if $x_{i+1}$ occurs twice in positive clauses, then it occurs exactly once in a negative clause by the $3$-bounded restriction.
In this case, we connect $q_i$ to $\overline{u_{i+1}}$. 

\begin{figure}[t]
\centering
\includegraphics[width=\textwidth]{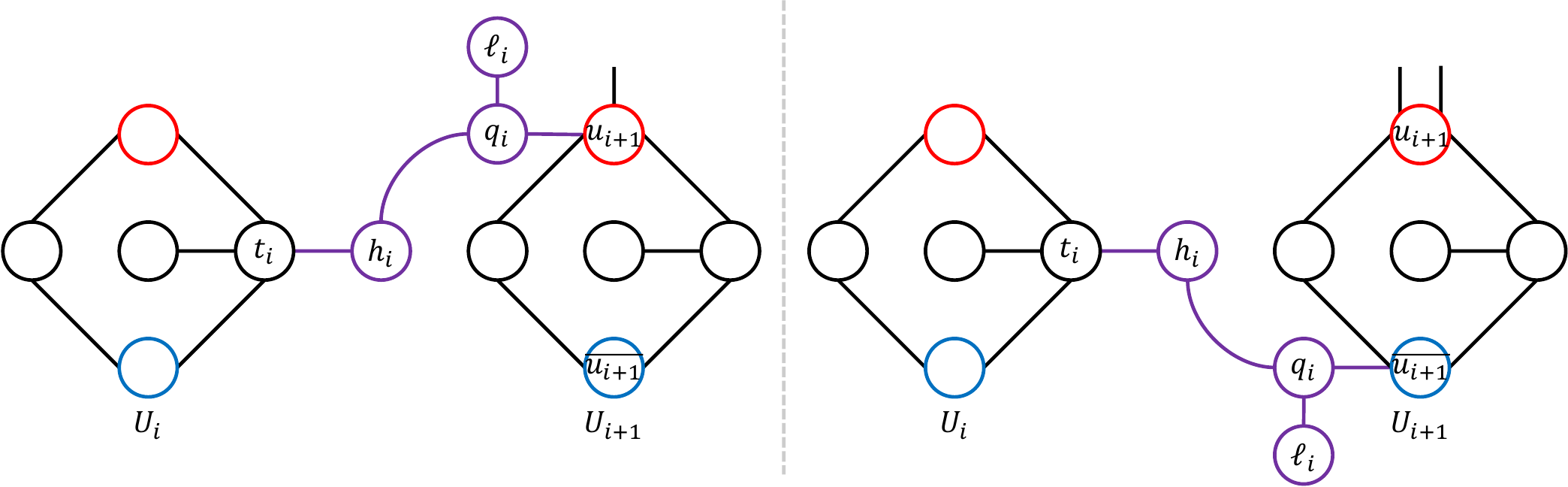}
\caption{A Connector gadget between variable gadgets. There are two ways to connect the connector gadget depending on the degree of $u_{i+1}$.} \label{fig:congad}
\end{figure}
\subsubsection{Overall Construction.}
Using the gadgets introduced above, we construct the graph $G$ for the instance of MCST from $\recrep$.
We first replace each vertex of $\recrep$ corresponding to $x_i\in\mathcal{X}$ with the variable gadget $U_i$.
Furthermore, for each $i \in [n-1]$, connect variable gadgets $U_i$ and $U_{i+1}$ using the connector gadget $H_i$.
Then, we introduce a vertex set $W = \{w_j\mid j\in[m]\}$ corresponding to clauses $\mathcal{C} = \{C_j\mid j \in[m]\}$, and add the edges $\{u_i w_j\mid x_i \in C_j\} \cup \{\overline{u_i}w_k\mid \overline{x_i}\in C_k\}$ for each $i \in [n]$.
\Cref{fig:reduction} shows the resulting graph $G$ constructed from $\recrep$ by the above procedure. 
Finally, we set $k=3n-1$, completing the construction of the instance $(G,k)$.

It is clear that the instance $(G,k)$ can be obtained in polynomial time.
We prove the correctness of our reduction. 
The following two lemmas complete the proof of Theorem \ref{the:planarhardness}.
\begin{lemma}[$\ast$]\label{lem:planarconstruction}
The constructed graph $G$ is planar, bipartite, and of maximum degree $4$.
\end{lemma}
\begin{lemma}[$\ast$]\label{lem:planarcorrectness}
The Boolean formula $\varphi$ is satisfiable if and only if the instance $(G,k)$ of {\mcst} has a solution.
\end{lemma}
\subsection{Unit disk graphs}
A graph $G$ is a \emph{unit disk graph} if each vertex of $G$ is associated with a (closed) disk with diameter $1$ on the plane, and two vertices $u,v\in V(G)$ are adjacent if and only if the two disks corresponding to $u$ and $v$ have non-empty intersection. 
In other words, $uv\in E(G)$ if and only if the Euclidean distance between the centers of disks corresponding to $u$ and $v$ is at most $1$.
Such a set of unit disks on the plane is called the \emph{geometric representation} of $G$.

\begin{theorem}
{\mcst} is NP-complete for unit disk graphs.
\label{the:unithardness}
\end{theorem}

To prove \Cref{the:unithardness}, we provide a polynomial-time reduction from {\mcst} on planar graphs with maximum degree $4$, which was shown to be NP-complete in \cref{the:planarhardness}, to {\mcst} on unit disk graphs.
Let $(G,k)$ be an instance of {\mcst} on planar graphs with maximum degree $4$.
We construct an instance $(G',k')$ of {\mcst} on unit disk graphs from $(G,k)$.

First, we embed $G$ in the plane using the following lemma. 

\begin{lemma}[\cite{10.5555/1963635.1963641}]
    A planar graph $G$ with maximum degree $4$ can be embedded in the plane using $O(|V(G)|)$ area in such a way that its vertices are at integer coordinates and its edges drawn so that they are made up of line segments of the form $x = i$ or $y = j$, for integers $i$ and $j$.
    \label{lem:xy}
\end{lemma}
\begin{figure}[t]
    \centering
    \includegraphics[width=\linewidth]{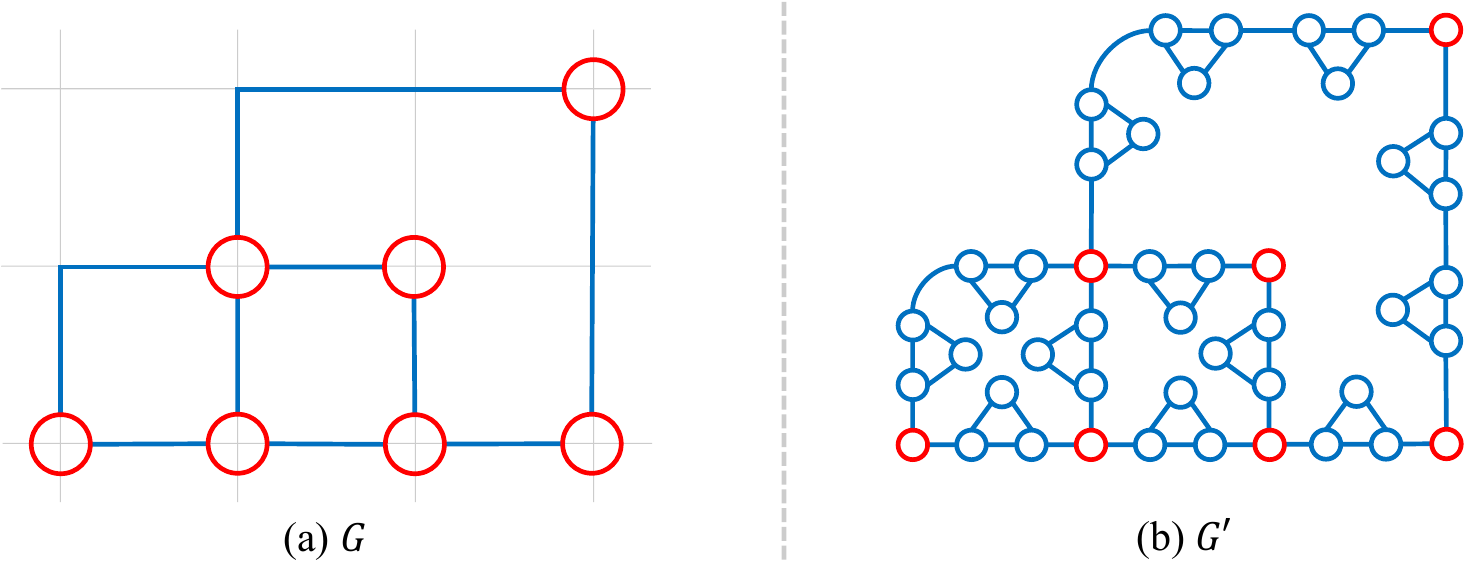}
    \caption{(a) A planar graph $G$ with maximum degree $4$ embedded according to \cref{lem:xy}, and (b) the unit disk graph $G'$ constructed from $G$.}
    \label{fig:unitdisk}
\end{figure}
\cref{fig:unitdisk}(a) illustrates an example of \Cref{lem:xy}.
This embedding allows us to represent each vertex $v_i \in V(G)$ with $i \in [n]$ by a coordinate $(x_i,y_i)$, where $x_i$ and $y_i$ are integers.
For each edge $v_iv_j\in E(G)$, we define the length function as $M(v_i,v_j)=|x_i-x_j|+|y_i -y_j|$, which corresponds to the Manhattan distance.
Using this length function, we construct the graph $G'$.
For each edge $uv$, let $\alpha = M(u,v)$, and introduce a set $\{\langle \ell_i, s_i, r_i \rangle \mid i \in [\alpha]\}$ of $\alpha$ copies of a 3-cycle.
Then, add the edges $u\ell_1$, $r_\alpha v$, and $r_i \ell_{i+1}$ for each $i \in [\alpha-1]$, which replaces the original edge $uv$ with a graph formed by connecting the 3-cycles.
\cref{fig:unitdisk}(b) depicts the resulting graph $G'$ constructed from $G$ by replacing each edge according to the above procedure.
Finally, we set $k' = k + \sum_{uv\in E(G)}M(u,v)$, completing the construction of the instance $(G',k')$.

It is clear that the instance $(G',k')$ can be obtained in polynomial time.
We prove the correctness of our reduction. 
The following two lemmas complete the proof of Theorem \ref{the:unithardness}.
\begin{lemma}[$\ast$] \label{lem:unit_correctness}
    The constructed graph $G'$ is a unit disk graph. 
\end{lemma}
\begin{lemma}[$\ast$] \label{lem:unit_correctness}
    $(G,k)$ has a solution to {\mcst} if and only if $(G',k')$ has a solution to {\mcst}. 
\end{lemma}
\section{Algorithms}\label{sec:alg}
\subsection{Clique-width}\label{subsec:cliquew}
In this section, we present an FPT algorithm parameterized by clique-width.
For a positive integer $w$, a \emph{$w$-labeled graph} is a graph whose vertices are labeled with integers from $[w] = \{1, 2, \dots, w\}$. 
We consider the following four operations on $w$-labeled graphs.\\\\
\begin{tabular}{@{}r p{0.7\textwidth}@{}}
\textbf{Introduce $i(v)$ :} & Create a graph consisting of a single vertex $v$ labeled with $ i \in [w] $. \\
\textbf{Union $\oplus$ :}     & Take the disjoint union of two $w$-labeled graphs. \\
\textbf{Relabel $\rho_{i \to j}$ :} & Replace every label $i$ with $j$. \\
\textbf{Join $\eta_{i, j}$ :} & Add all possible edges between vertices labeled $i$ and $j$. \\
\end{tabular}\\\\
The \emph{clique-width} of a graph $G$ is the smallest integer $w$ such that a $w$-labeled graph isomorphic to $G$ can be obtained from a sequence of these operations.
The sequence can be represented by a rooted tree, called a \emph{$w$-expression tree} $\decom$ of $G$, whose nodes are classified according to the above operations, namely, \emph{introduce nodes}, \emph{union nodes}, \emph{relabel nodes}, and \emph{join nodes}, respectively.
It is known that, given a graph $G$ of clique-width at most $w$, there exists an algorithm that computes a $(2^{w + 1} - 1)$-expression tree for $G$ in time $O(|V (G)|^ 3)$~\cite{10.1007/978-3-540-75520-3_16,OUM2006514,10.1007/11604686_5}.

First, we define a problem equivalent to \mcst\ and prove that a solution to the problem can be expressed in monadic second-order logic, called MSO$_1$.
To show this, we demonstrate the following equivalence.
\begin{lemma}[$\ast$]\label{lem:cut}
    Let $G=(V, E)$ be a connected graph with at least two vertices, and let $S$ be a vertex subset of $G$.
    Then the following claims (a) and (b) are equivalent:
    \begin{enumerate}[(a)]
        \item there exists a spanning tree $T$ of $G$ covered by $S$; and 
        \item for any non-empty vertex subset $C\subset V$, there exist two vertices $u$ and $v$ that satisfy the following four conditions: $\mathrm{(1)}$ $uv\in E$; $\mathrm{(2)}$ $u\in C$; $\mathrm{(3)}$ $v\in V\setminus C$; and $\mathrm{(4)}$ $u\in S$ or $v\in S$
    \end{enumerate}
\end{lemma}
A vertex set $S$ that satisfies the property~(b) in \Cref{lem:cut} can be expressed in MSO$_1$. 
The result of Courcelle et al.\ implies that the minimization problem of $S$, as well as {\mcst}, admits an FPT algorithm parameterized by clique-width~\cite{Courcelletheorem}.
However, it is known that its running time can take the form of an exponential tower of $|\varphi|$, where $|\varphi|$ is the length of an MSO$_1$ formula $\varphi$~\cite{FrickG04}.
Therefore, we will design a faster FPT algorithm.
Our main theorem of this section is as follows.
\begin{theorem} \label{the:alg_cliquewidth}
Given a $w$-expression tree of an $n$-vertex graph, {\mcst} can be solved in time $ 2^{2^{O(w)}}n$.
\end{theorem}

To show \Cref{the:alg_cliquewidth}, we present an algorithm using dynamic programming over a given $w$-expression tree $\decom$.
Let $t$ be a node in $\decom$.
We denote by $G_t$ the labeled graph constructed by operations in the subtree rooted at $t$.
We compute candidate solutions on $G_t$ for each node $t$ of $\decom$ in a bottom-up manner.
Recall that, by Observation~\ref{obs:boundary}, there exists a spanning tree of $G$ covered by a vertex set $S$ if and only if the boundary subgraph for $S$ is connected.
Thus, we will store candidate solutions as boundary subgraphs instead of spanning trees.

Consider the boundary subgraph $H^\ast$ for a vertex subset $S^\ast$ of $G$.
The restriction of $S^\ast$ to the subgraph $G_t$ yields the boundary subgraph $H$ for $S = V(G_t) \cap S^\ast$.
Hence, we guess $S$ at each node $t$.
Moreover, since $H$ is required to be connected at the root of $\decom$, we also need to record the connected components of $H$.
To obtain an FPT algorithm for clique-width, these pieces of information are managed using labels from $[w]$ instead of vertex sets.
Let $[w]$ be the set of labels, and let $\pow = 2^{[w]}$ denote its power set.
For the boundary subgraph $H$ for a vertex set $S$ of $G_t$ and a pair $(C, X) \in \pow \times \pow$, we say that a component $H'$ of $H$ \emph{matches} $(C,X)$ if 
\begin{itemize}
    \item the label set from $V(H')$ are exactly $C$; and 
    \item the label set from $V(H')\cap S$ are exactly $X$.
\end{itemize}
A function $f:\pow \times \pow \to \{0,1,2\}$ is \emph{valid} if there exists the boundary subgraph $H$ for a vertex set $S$ of $G_t$ such that every pair $(C, X) \in \pow \times \pow$ satisfies $f(C,X) = \min \{2,h'\}$, where $h'$ is the number of components $H'$ of $H$.
Otherwise, $f$ is called \emph{invalid}.
Note that we do not distinguish cases where there are at least two such components.
We also note that whenever a component $H'$ matches $(C, X)$, we have $X \subseteq C$.
Consequently, it is not necessary to consider a pair $(C, X)$ with $X \setminus C \neq \emptyset$.

For each node $t$ of $\decom$, we define $\mathrm{dp}_{t}$ as a DP table that, for each valid function $f:\pow \times \pow \to \{0,1,2\}$, stores the minimum size of vertex sets $S$ such that the number of components in the boundary subgraph for $S$ corresponds to $f$.
For an invalid function $f$, we set $\mathrm{dp}_{t}(f) = +\infty$.
The number of possible pairs $(C, X) \in \pow \times \pow$ is bounded by $4^w$.
The function $f$ independently assigns one of three values $\{0,1,2\}$ to each pair $(C, X)$, and hence the number of possible functions is $3^{4^w}$.
Therefore, the size of the DP table is $2^{2^{O(w)}}$ for each node $t$.

At the root node $r$ of $\decom$, the desired boundary subgraph $H$ must be connected.
Therefore, we look up $\mathrm{dp}_{r}(f)$ for every function $f: \pow \times \pow \to \{0,1,2\}$ such that $f(C, X) = 1$ for exactly one pair $(C,X) \in \pow \times \pow$ and $f(C', X') = 0$ for every other pair $(C',X')$.
Let $F$ be the set of functions $f$ satisfying the above conditions.
An instance $(G,k)$ of \mcst\ is a yes-instance if and only if  $\min\{\mathrm{dp}_r(f) \mid f \in F\} \le k$.
\subsection{Interval graphs} \label{sec:interval} 
This section shows the following theorem.

\begin{restatable}{theorem}{theoreminterval} 
     {\mcst} can be solved in linear time for interval graphs.
    \label{the:interval}
\end{restatable}

A sequence $(v_1, v_2, \ldots, v_n)$ of the vertices in an $n$-vertex graph $G$ is an \emph{interval ordering} $I$ if it satisfies the following condition: for any three integers $a,b,c\in [n]$ with $a<b<c$, if $v_av_c\in E(G)$, then $v_bv_c\in E(G)$.
A graph is an \emph{interval graph} if it has an interval ordering.
An interval ordering of a given interval graph can be computed in linear time~\cite{HsuW99}.
We give \cref{code:greedy} to solve {\mcst} for interval graphs as the proof of \Cref{the:interval}.
\begin{algorithm}[th]
	\caption{Solving \mcst\ on interval graphs}
	\label{code:greedy}
	\begin{algorithmic}[1]
    \Input An interval ordering $I = (v_1,v_2,\ldots, v_n)$ of a connected interval graph $G$
    \Output A minimum vertex cover $S$ among spanning trees of $G$
	\State $S\leftarrow\emptyset$, $V_T \leftarrow \{v_1\}$ 
		\While{$V_T\ne V(G)$} 
            \State $t_1\leftarrow\min\{i\mid v_i\notin V_T\}$
            \State $t_2\leftarrow \max\{i\mid v_i\in V_T\}$
            \State $t\leftarrow\min\{t_1,t_2\}$
			\State $s \leftarrow \max \{i \mid v_i \in N[v_t]\}$
			\State $S\leftarrow S \cup \{v_s\}$
			\State $V_T \leftarrow V_T \cup N[v_s]$
		\EndWhile 
		\State \Return $S$
	\end{algorithmic}
\end{algorithm}

Observe that \cref{code:greedy} can be implemented to run in linear time. 
The following lemma completes the proof of \cref{the:interval}.
\begin{lemma}[$\ast$]\label{lem:interval}
    \cref{code:greedy} outputs a minimum vertex cover among spanning trees of a connected interval graph.
\end{lemma}
\vspace{1em}
\noindent \textbf{Acknowledgments.} We thank anonymous reviewers for their valuable comments and suggestions which greatly helped to improve the presentation of this paper.

 \bibliographystyle{plainurl}
 \bibliography{reference}

\newpage
\appendix
\section{Omitted proofs in \Cref{sec:hard}}
\subsection{Proof of \Cref{lem:planarconstruction}}\label{app:lem:planarconstruction}
\begin{proof}
Obviously, the variable gadgets and connector gadgets are planar and bipartite.
It is not hard to see that $G$ is also bipartite.
Moreover, since $\recrep$ is drawn according to a rectilinear representation, the gadgets do not introduce any edge crossings.
Hence, $G$ is planar.

Observe that each vertex of variable gadgets and connector gadgets has degree at most~$4$.
Since $\varphi$ is a 3-CNF formula, each vertex in $W$ has degree at most~$3$.
Moreover, the 3-bounded restriction of $\varphi$ ensures that each of the vertices $u_i$ and $\overline{u_i}$ in the variable gadget $U_i$ is adjacent to at most two vertices in $W$.
By the construction of $G$, the connector gadget introduces an edge from $q_i$ to the vertex among $\{u_{i+1},\overline{u_{i+1}}\}$ that has fewer neighbors in $W$.
Therefore, the maximum degree of $G$ is at most~$4$.
\end{proof}

\subsection{Proof of \Cref{lem:planarcorrectness}}\label{app:lem:planarcorrectness}
\begin{proof}
    (Only-if direction) 
Suppose that $\varphi$ is satisfiable.
Let $A:\mathcal{X} \rightarrow{\{\text{True}, \text{False}\}}$ be a satisfying assignment of $\varphi$.
If $A(x_i) = \mathrm{True}$ (resp.\ $A(x_i) = \mathrm{False}$), let $C_i = \{ t_i, u_i \}$ (resp.\ $C_i = \{ t_i, \overline{u_i} \}$) and let $C = (\bigcup_{i\in[n]} C_i) \cup \{q_k \mid k \in [n-1]\}$.
Observe that the boundary subgraph for $C$  is connected in each variable gadget and connector gadget, and $|C|=3n-1$.
Since $\varphi$ is satisfiable and $C_i$ is defined according to the satisfying assignment $A$ of $\varphi$, for each $j \in [m]$, the vertex $w_j \in W$ is adjacent to some vertex in $C$.
Therefore, the boundary subgraph for $C$ of $G$ is connected, and $C$ produces a solution to $(G,3n-1)$ from Observation~\ref{obs:boundary}.

(If direction)
By \Cref{obs:boundary}, we may suppose that there exists a vertex subset $C$ of $G$ such that $|C| \le 3k-1$ and the boundary subgraph $H$ for $C$ is connected.
We also suppose that $C$ contains the vertex $t_i$ for each $i \in [n]$ and the vertex $q_i$ for $i \in [n-1]$ because the vertex $s_i$ and the vertex $\ell_i$ have degree~$1$.
It should be noted that one of $u_i$, $\overline{u_i}$, and $r_i$ must be contained in $C$; otherwise, $H$ would not be connected.
Thus, since $|C| \le 3n-1$, no vertex $w_j \in W$ is contained in $C$.
This implies that for each $w_j \in W$, there exists an integer $i_j \in [n]$ such that either $u_{i_j}$ or $\overline{u_{i_j}}$ is in $C$ and adjacent to $w_j$ on $H$.
We define a truth assignment $A: \mathcal{X} \to \{\mathrm{True}, \mathrm{False}\}$ such that for each $i \in [n]$, $A(x_i) = \mathrm{True}$ (resp.\ $A(x_i) = \mathrm{False}$)  if $u_i \in C$ (resp.\ $\overline{u_i} \in C$).
If $r_i \in C$, then we arbitrarily assign either True or False to $x_i$.
Consequently, $A$ satisfies the Boolean formula $\varphi$.
\end{proof}

\subsection{The embedding geometric representation in \Cref{the:unithardness}}\label{app:unitconstruction}
Let $\mathbb{Z}^2 = \{(a,b) \mid a,b \in \mathbb{Z}\}$ be the integer lattice, and let $G$ be a graph embedded in the plane whose vertices lie on $\mathbb{Z}^2$ and edges are drawn as in \cref{lem:xy}.
Scaling $\mathbb{Z}^2$ to the integer lattice $2\mathbb{Z}^2$, we obtain the graph $G^2$, isomorphic to $G$, in which each line segment has even length.
Recall that our reduction introduced one $3$-cycle per unit of Manhattan distance of each edge in $G$ to construct $G'$.
Here, we insert one $3$-cycle per \emph{two units} of Manhattan distance.
Due to the scaling to $2\mathbb{Z}^2$, this does not change the resulting graph $G'$.

We classify line segments consisting of edges in $G^2$ into the following three cases:
    \begin{itemize}
        \item[(a)] two vertices of $G^2$ are placed at both $(a,b)$ and $(a+2,b)$;
        \item[(b)] exactly one vertex of $G^2$ is placed at $(a,b)$ or $(a+2,b)$; and 
        \item[(c)] no vertex is placed at $(a,b)$ or $(a+2,b)$, 
    \end{itemize}
where $(a,b), (a+2,b)\in 2\mathbb{Z}^2$ and a line segment joins them.
\cref{fig:diskdist} illustrates how to replace a line segment of length~$2$ with a $3$-cycle and how to provide the corresponding geometric representation.
Consider the square $WXYZ$ with corners $W = (a,b)$, $X = (a+1,b+1)$, $Y = (a+2, b)$, and $Z = (a+1, b-1)$.

\begin{figure}[t]
    \centering
    \includegraphics[width=\linewidth]{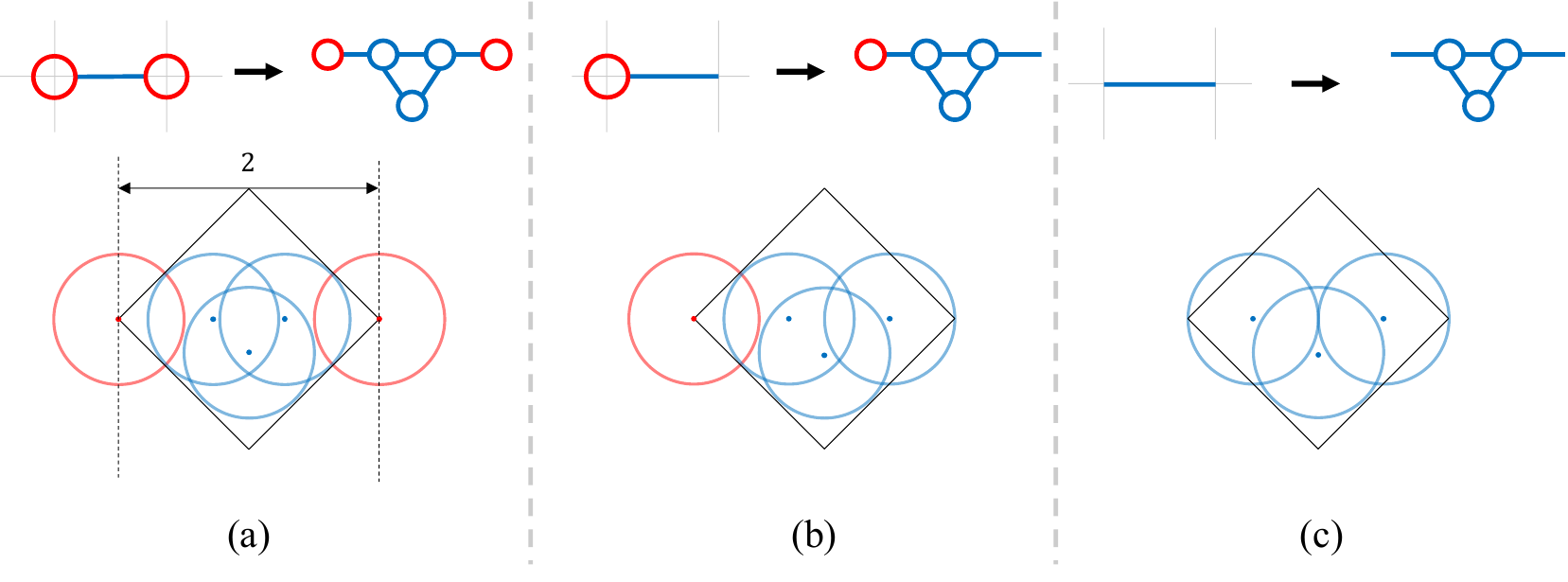}
    \caption{The 3-cycles replacing line segments of length $2$ and the corresponding geometric representations for the three cases: (a) two vertices are placed at both integer coordinates; (b) exactly one vertex is placed; and (c) no vertex is placed.}
    \label{fig:diskdist}
\end{figure}
\begin{figure}[t]
    \centering
    \includegraphics[width=\linewidth]{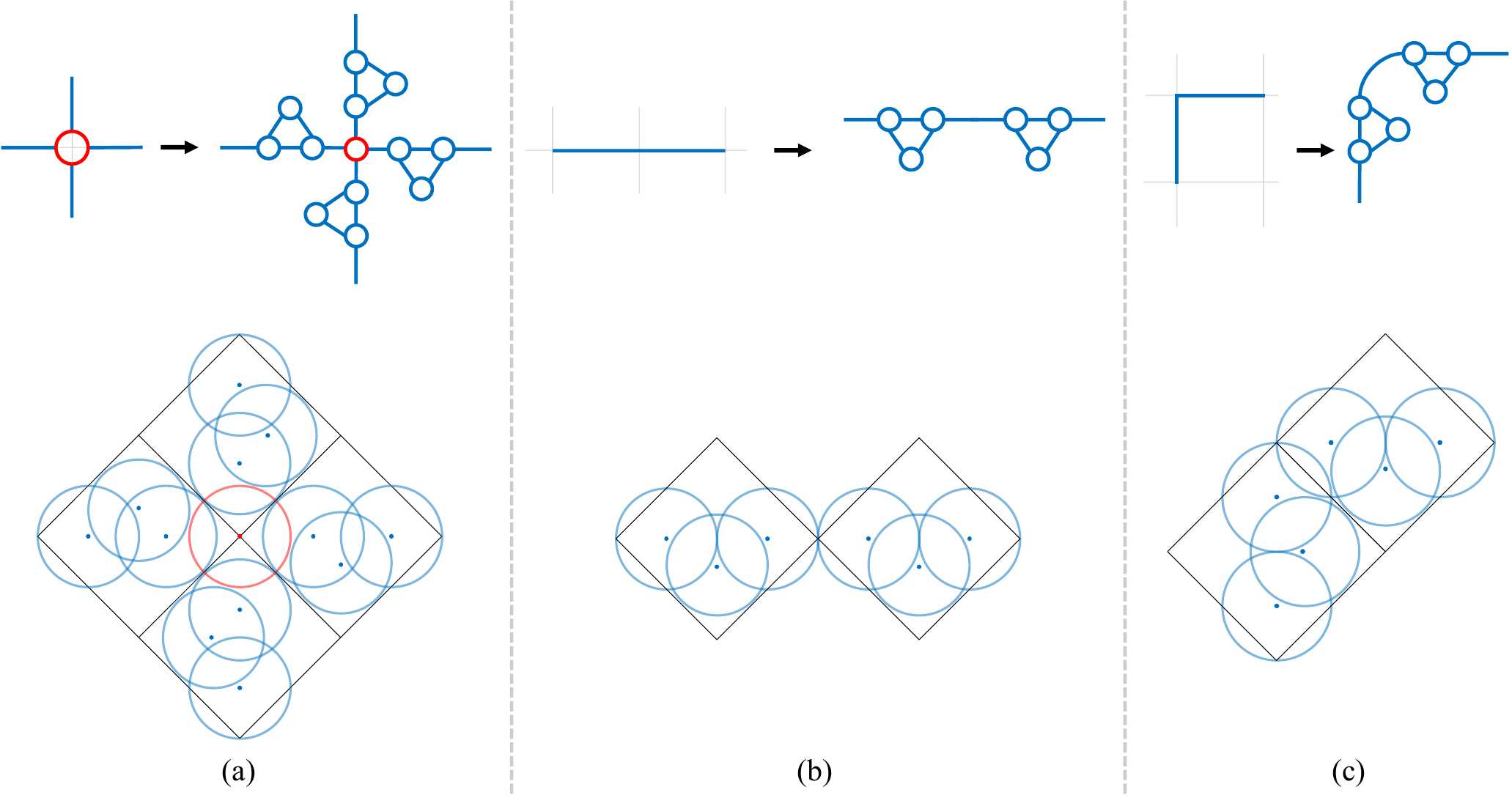}
    \caption{The 3-cycles replacing line segments and the corresponding geometric representations for the three cases: (a) four line segments incident to a vertex of degree~$4$; (b) two line segments drawn straight; and (c) two line segments making a right-angle turn. In any case, no undesired intersections occur.}
    \label{fig:georep}
\end{figure}

When replacing a line segment of type (a), which corresponds to an edge $uv$ in $G^2$, with a $3$-cycle $\langle \ell, s, r\rangle$, we add five disks $D_u$, $D_\ell$, $D_s$, $D_r$, and $D_v$ with diameter~$1$ (see \cref{fig:diskdist}(a)).
The centers of $D_u$ and $D_v$ are located at points $W$ and $Y$, respectively.
Furthermore, the centers of $D_\ell$, $D_s$, and $D_r$ are located at the points $(a+\sqrt2/2+\varepsilon, b)$, $(a+1, b-1+\sqrt2/2+\varepsilon)$, and $(a+2-\sqrt2/2-\varepsilon, b)$, respectively, where $\varepsilon$ is a sufficiently small positive number.
Observe that the $3$-cycle $\langle \ell, s, r\rangle$ together with the edges $u\ell$ and $rv$ is represented by these disks.
The key point is that the disks $D_\ell$, $D_s$, and $D_r$ are properly contained in the square $WXYZ$, and hence they do not intersect with disks that represent other $3$-cycles.

Consider a line segment of type (b).
Without loss of generality, assume that a vertex $u$ of $G^2$ is located at point $A$.
We add four disks $D_u$, $D_\ell$, $D_s$, and $D_r$ with diameter~$1$ (see \cref{fig:diskdist}(b)).
The disks $D_u$, $D_\ell$, and $D_s$ are defined in the same way as in case (a).
In contrast, the center of $D_r$ is located at the point $(a+1.5, b)$.
The disks $D_\ell$ and $D_s$ are properly contained in the square $WXYZ$, whereas the disk $D_r$ contains the point $(a+2,b)$ and intersects with a disk in the next $3$-cycle.

For a line segment of type (c), the centers of three disks $D_\ell$, $D_s$, and $D_r$ are located at the points $(a+0.5, b)$, $(a+1, b-1+\sqrt2/2+\varepsilon)$, and $(a+1.5, b)$, respectively (see \cref{fig:diskdist}(c)).
In this case, $D_\ell$ and $D_r$ intersect disks in the previous and next cycles.

In the case where a line segment joins the points $(a,b), (a,b+2) \in 2\mathbb{Z}^2$, we also define disks in the same way.

We replace all line segments in $G^2$ with disks as above.
The disks are said to have \emph{undesired intersections} if it does not yield the geometric representation of $G'$.
We show that no undesired intersections among the disks arise.
Observe that, if such undesired intersections occur, $3$-cycles are involved.
Consider three disks $D_\ell$, $D_s$, and $D_r$ that correspond to a $3$-cycle $\langle \ell, s, r \rangle$ in $G'$.
We may assume that the center of $D_\ell$ is located at the point $(a+0.5, b)$ and hence it contains the point $(a,b) \in 2\mathbb{Z}^2$.
Otherwise, as mentioned earlier, there are no undesired intersections involving $D_\ell$ (see also \cref{fig:georep}(a)).
Note that $D_\ell$ means the existence of the line segment $L_1$ from $(a,b)$ to $(a+2,b)$ that is part of an edge in $G^2$. 

In the case where there exists the line segment from $(a-2,b)$ to $(a,b)$, it is straightforward that no undesired intersections occur (see \cref{fig:georep}(b)).
Note that there exists no disk whose center is located in the square $WXYZ$ with corners $W = (a,b)$, $X =(a+1,b+1)$, $Y = (a,b+2)$, and $Z = (a-1,b+1)$; otherwise, this would imply that three line segments extend from the point $(a,b)$ or there exists a line segment on the line $b+1$, which contradicts the embedding of $G^2$. 
This claim is also applicable to the square with corners $(a,b)$, $(a+1,b-1)$, $(a,b-2)$, and $(a-1,b-1)$.

Let $L_2$ be the line segment from $(a,b)$ to $(a,b-2)$ and suppose that it corresponds to a $3$-cycle $\langle \ell', s', r' \rangle$ in $G'$.
The two line segments $L_1$ and $L_2$ make a right-angle turn (see \cref{fig:georep}(c)).
Since $D_{s'}$ and $D_{\ell'}$ are located at the points $(a + 1 - \sqrt2/2 - \varepsilon, b-1)$ and $(a, b-1.5)$, the Euclidean distance between $D_\ell$ and $D_{s'}$, as well as that between $D_\ell$ and $D_{\ell'}$, is greater than~$1$.
In the same argument previously mentioned, we also remark that there exists no disk whose center is located in the square with corners $(a,b)$, $(a-1,b+1)$, $(a-2,b)$, and $(a-1,b-1)$, and the square with corners $(a,b)$, $(a+1,b+1)$, $(a,b+2)$, and $(a-1,b+1)$.
Therefore, there are no unintended intersections involving $D_\ell$, and the same argument applies to all disks.
In conclusion, $G'$ is a unit disk graph.

\subsubsection{Proof of \Cref{lem:unit_correctness}} \label{app:unit_correctness}
\begin{proof}
    Let $G_0 = G$.
For an integer $i \in [m]$, where $m$ is the number of edges in $G$, we denote by $G_i$ the graph obtained from $G_{i-1}$ by replacing an edge $uv \in V(G)$ with a sequence of 3-cycles, that is, the graph $G_{uv}$ such that $V(G_{uv})=\{u,v\} \cup \{\langle \ell_i,s_i,r_i\rangle \mid i \in [\alpha]\}$ and $E(G_{uv}) =\{u\ell_1,r_\alpha v\}\cup\{r_i \ell_{i+1}\mid i \in [\alpha-1]\}$, where $\alpha = M(u,v)$.
We prove the following two claims: (a) if there exists a vertex subset $S_{i-1}$ of $G_{i-1}$ such that the boundary subgraph for $S_{i-1}$ is connected, then there exists a vertex subset $S_{i}$ of $G_{i}$ such that $|S_{i}| \le |S_{i-1}| + M(u,v)$ and the boundary subgraph for $S_{i}$ is connected; and (b) if there exists a vertex subset $S_{i}$ of $G_{i}$ such that the boundary subgraph for $S_{i}$ is connected, then there exists a vertex subset $S_{i-1}$ of $G_{i-1}$ such that $|S_{i-1}| \le |S_{i}| - M(u,v)$ and the boundary subgraph for $S_{i-1}$ is connected.
Since $G_0 = G$, $G_m = G'$, and $k' = k + \sum_{uv\in E(G)}M(u,v)$, \Cref{lem:unit_correctness} follows from these claims and \Cref{obs:boundary}.

(The proof of claim (a).)
Suppose that there exists a vertex subset $S_{i-1}$ of $G_{i-1}$ such that the boundary subgraph for $S_{i-1}$ is connected.
Let $\alpha = M(u,v)$, $L = \{\ell_i \mid i \in [\alpha]\}$, and $R = \{r_i \mid i \in [\alpha]\}$.
We define $S_{uv}$ according to whether $u$ and $v$ belong to $S_{i-1}$ as follows:
    \begin{align}
    S_{uv}=
    \begin{cases}
        \{u,v\} \cup R & \text{ if } u,v \in S_{i-1}\\
        \{u\} \cup R & \text{ if } u \in S_{i-1} \text{ and } v \notin S_{i-1}\\
        \{v\} \cup L& \text{ if } u \notin S_{i-1} \text{ and } v \in S_{i-1}\\
        R & \text{ otherwise } 
    \end{cases}
    \label{align:Guv}
\end{align}
Let $S_i = S_{i-1} \cup S_{uv}$.
Clearly, $|S_i| \leq  |S_{i-1}| + M(u,v)$.
We show that the boundary subgraph $H_{i}$ for $S_i$ is connected in $G_{i}$.

If $uv \in E(H_{i-1})$, then we have $u \in S_i$ or $v \in S_i$ from the definition of $H_{i-1}$.
By \cref{align:Guv}, the boundary subgraph for $S_i$ contains $G_{uv}$ as a subgraph.
Thus, the boundary subgraph $H_{i}$ for $S_i$ is also connected in $G_{i}$.
If $uv \notin E(H_{i-1})$, then there exists a vertex $p \in  N_{G_{i-1}}(v) \cap S_{i-1}$.
By \cref{align:Guv}, the boundary subgraph for $S_{uv}$ contains the vertices of $G_{uv}$ except for the vertex $u$.
Since $v$ is adjacent to $p$ in $H_{i-1}$, the boundary subgraph $H_{i}$ of $G_{i}$ for $S_i$ is also connected.

(The proof of claim (b).)
Suppose that there exists a vertex subset $S_{i}$ of $G_{i}$ such that the boundary subgraph for $S_{i}$ is connected.
Let $S_{uv} = V(G_{uv})\cap S_{i} \setminus \{ u,v \}$ and $\alpha = M(u,v)$.
For each 3-cycle $\langle \ell_i, s_i, r_i\rangle$ with $i \in [\alpha]$, at least one of $\ell_i, r_i$ and $s_i$ must be contained in $S'_{uv}$ because $s_i$ is only adjacent to $\ell_i$ and $r_i$.
Therefore, $|S_{uv}| \ge M(u,v)$ holds.
Our task is to show that there exists a vertex subset $S_{i-1}$ of $G_{i-1}$ such that $|S_{i-1}| \le |S_{i}| - M(u,v)$ and the boundary subgraph for $S_{i-1}$ is connected in $G_{i-1}$.

In the case that $u \in S_{i}$ or $v \in S_{i}$, observe that $S_{i-1} = S_{i} \setminus S_{uv}$ is a desired vertex set.
Let $R = \{r_i \mid i \in [\alpha]\}$.
If $|S_{uv}| > \alpha$, then $S'_{i} = (S_{i} \setminus S_{uv}) \cup (\{u \} \cup R)$ satisfies $|S'_{i}| \le |S_{i}|$, and hence this can be reduced to the previous case.
Suppose that $u \notin S_{i}$, $v \notin S_{i}$, and $|S_{uv}| = \alpha$.
Then, at least one of $\ell_1$ and $r_{\alpha}$ belongs to $S_{i}$; otherwise, $H_{i}$ would not be connected.
Without loss of generality, assume that $r_\alpha \in S_{i}$.
Observe that $H_{i}$ is connected only when $R \subseteq S_{i}$.
This implies that $\ell_1 \notin S_{i}$ and hence $u \ell_1 \notin E(H_{i})$.
In other words, there is a path in $H_{i}$ between any vertices $x$ and $y$ that avoids $u \ell_1$.
Therefore, the boundary subgraph for $S_{i-1} = S_{i} \setminus S_{uv}$ is connected in $G_{i-1}$ and $|S_{i-1}| = |S_{i}| - M(u,v)$ holds.
This completes the proof.
\end{proof}

\section{Omitted discussions in \Cref{sec:alg}}
\subsection{Monadic second-order logic on graphs} \label{app:MSO}


We define the monadic second-order logic dealt with in this paper, MSO$_1$, referring to~\cite{CyganFKLMPPS15}.
MSO$_1$ formulas employ two types of variables: vertex variables and vertex-set variables.
They are defined inductively from smaller subformulas.
We begin by introducing the most basic components, called \emph{atomic formulas}: 
\begin{itemize}
    \item $x = y$ for two variables $x$ and $y$ of the same type;
    \item $u \in X$ for a vertex variable $u$ and a vertex-set variable $X$; and 
    \item $\adj{u}{v}$ for two vertex variables $u$ and $v$, which is evaluated as True if and only if $u$ and $v$ are adjacent in a graph.
\end{itemize}
MSO$_1$ formulas also use standard Boolean operators ($\neg$, $\land$, $\lor$, $\Rightarrow$) and quantifications of a variable ($\forall u$, $\exists u$, $\forall X$, $\exists X$).
For MSO$_1$ formulas $\varphi$ and $\varphi'$, the following are also MSO$_1$ formulas: $\neg \varphi$, $\varphi \land \varphi'$, $\varphi \lor \varphi'$, $\varphi \Rightarrow \varphi'$, $\forall u: \varphi$, $\exists u: \varphi$, $\forall X: \varphi$, and $\exists X: \varphi$

Recall that, given a connected graph $G = (V, E)$, the minimization variant of \mcst\ can be restated as the following problem: find a minimum set $S \subseteq V$ such that, for any non-empty vertex subset $C\subset V$, there exist two vertices $u$ and $v$ that satisfy the following four conditions: (1) $uv \in E$; (2) $u\in C$; (3) $v\in V\setminus C$; and (4) $u\in S$ or $v\in S$.
The property that a vertex set $S \subseteq V$ is a solution to the restated problem can be expressed in MSO$_1$ as follows:
\begin{align*}
    \forall C~ \exists u,v:~ & \neg(C = \emptyset) \land \neg (C = V) \\
    & \Rightarrow \adj{u}{v} \land u \in C \land \neg (v \in  C) \land (u \in S \lor v \in S).
\end{align*}

Thus, the problem of minimizing $S$ admits an FPT algorithm parameterized by the clique-width of a given graph $G$ as a direct consequence of the following theorem.

\begin{theorem}[\cite{Courcelletheorem}]
For a graph $G$ with $n$ vertices and clique-width $\cw$, every minimization or maximization problem whose solution is expressible in MSO$_1$ can be solved in time $f(|\varphi|, \cw)\cdot n$, where $|\varphi|$ is the length of an MSO$_1$ formula that expresses a solution to the problem and $f$ is some computable function.
\end{theorem}

\subsection{Proof of \Cref{lem:cut}}
\begin{proof}
    {[(a) $\Rightarrow$ (b)]} Assume that there is a spanning tree $T$ of $G$ covered by $S$.
    Consider a partition $(C, V\setminus C)$ of $V$ such that $C \neq \emptyset$ and $C \neq V$.
    Since $T$ is a spanning tree of $G$, there is an edge $uv\in E(T)$ such that $u\in C$ and $v\in V\setminus C$. 
    From the fact that $T$ is covered by $S$, we conclude that $u\in S$ or $v\in S$ holds.

    {[(b) $\Rightarrow$ (a)]} We prove this claim by constructing a spanning tree $T$ covered by $S$ incrementally.
    Initially, suppose that $T$ consists of an arbitrary single vertex in $G$.
    It is clear that $T$ is covered by $V(T) \cap S$ (even though $V(T) \cap S$ could be empty).
    Consider a partition $(V(T), V \setminus V(T))$.
    From the assumption, there is an edge $uv \in E$ such that $u \in V(T)$, $v \in V \setminus V(T)$, and $u \in S$ or $v \in S$ holds.
    Let $T'$ be a subgraph of $G$ such that $V(T') = V(T) \cup \{v\}$ and $E(T') = E(T) \cup \{ uv \}$.
    Observe that $T'$ forms a tree and is covered by $V(T') \cap S$.
    Iterating this procedure, we can construct a spanning tree of $G$ covered by $S$.
\end{proof}

\subsection{Dynamic programming in \Cref{subsec:cliquew}}\label{app:subsec:cliquew}
In the following, we demonstrate how to compute $\mathrm{dp}_{t}$ according to the type of a node $t$ in $\decom$.

\medskip\noindent\textbf{Introduce node.}
When a vertex $v$ with label $i$ is introduced, there are exactly two possible candidate solutions: the boundary subgraph for $\{v\}$ and the boundary subgraph for $\emptyset$.
Therefore, there are only two valid functions $f_1$ and $f_2$ as follows: for each $(C,X)\in \pow \times \pow$
    \begin{align*}
        f_1(C,X)=
        \begin{cases}
            1 & \text{if $(C,X)=(\{i\},\{i\})$}\\
            0 & \text{otherwise}
        \end{cases}
    , \quad
        f_2(C,X)=
        \begin{cases}
            1 & \text{if $(C,X)=(\{i\},\emptyset)$}\\
            0 & \text{otherwise}
        \end{cases}
    \end{align*}
All other functions are invalid.
Therefore, the entry $\text{dp}_{t}$ of each cell in DP table is defined as follows.
    \begin{align*}
        \text{dp}_{t}(f)=
        \begin{cases}
            1 & \text{if $f=f_1$}\\
            0 & \text{if $f=f_2$}\\
            +\infty & \text{otherwise}
        \end{cases}
    \end{align*}
\noindent\textbf{Union node.}
Let $\ell$ and $r$ be the two child nodes of $t$.
The union of boundary subgraphs for $G_\ell$ and $G_r$ yields a boundary subgraph for $G_t$.
Thus, a function $f$ is valid if there exist valid functions $f_\ell$ and $f_r$ that satisfy the following conditions.
    \begin{align*}
            \forall (C,X) \in \pow \times \pow: f (C,X)=\text{min}\{2,f_\ell(C,X)+f_r(C,X)\}
    \end{align*}
Let $M_f$ be the set of all pairs $(f_\ell, f_r)$ of functions satisfying the above condition.
The entry $\text{dp}_{t}$ is defined as follows:
    \begin{align*}
        \text{dp}_t(f)=
        \begin{cases}
        \min\{\text{dp}_\ell(f_\ell)+\text{dp}_{r}(f_r) \mid (f_\ell,f_r)\in M_{f}\}& \text{ if }M_f\ne \emptyset\\
        +\infty & \text{ otherwise }
        \end{cases}
    \end{align*}
\noindent\textbf{Relabel node.}
Suppose that node $t$ has a child node $t'$ and the vertices with label $i$ are relabeled to $j$.
Consider components of a subgraph of $G_t$.
Since $G_t$ has no vertex with label $i$, there is no component that matches $(C, X) \in \pow \times \pow$ with $i \in C$. 
Furthermore, for $(C, X)$ with $i,j \notin C$, every component of $G_{t}$ that matches $(C, X)$ is also a component of $G_{t'}$ that matches $(C, X)$.
The nontrivial case is when $(C, X)$ satisfies $i \notin C$ and $j \in C$.
Let $H'$ be a component of $G_t$ that matches $(C, X)$.
We have to pay attention that $H'$ may not match $(C, X)$ at node $t'$ due to the relabeling.
For a set $A \in \pow$, let $A_{i \to j} =  A \cup \{j\} \setminus \{i\}$ if $i \in A$; otherwise, let $A_{i \to j} = A$.
We denote by $\mathcal{A}(C, X)$ the set of pairs $(C', X') \in \pow \times \pow$ that corresponds to $(C, X)$ after the relabeling.
Formally, 
\begin{align*}
    \mathcal{A}(C,X)= \{ (C',X') \in \pow \times \pow \mid C = C'_{i\to j}, X = X'_{i \to j}  \}.
\end{align*}
After relabeling, the number of components that match $(C, X)$ equals the sum of the number of components that match pairs in $\mathcal{A}(C, X)$. 
Therefore, a function $f$ is valid if there exists a valid function $f'$ at node $t'$ such that the following equation holds for every $(C,X)\in \pow \times \pow$.
\begin{align*}
     f(C,X)=\begin{cases}
        0 & \text{if  $i\in C$}\\
        f'(C,X)& \text{if  $i\notin C$ and $j\notin C$} \\
        \min\{2,\sum_{(C^{\prime},X^{\prime})\in \mathcal{A}(C,X)}f'(C^{\prime},X^{\prime})\}& \text{otherwise} \\
    \end{cases}
\end{align*}

To update $\mathrm{dp}_t(f)$, we consider all functions $f'$ that hold the above condition.
Let $M_f$ be the set of all such functions $f'$. 
The entry $\text{dp}_{t}$ is: 
\begin{align*}
        \text{dp}_t(f)=
        \begin{cases}
        \min\{\text{dp}_{t'}(f')\mid f'\in M_f\} & \text{if $M_f \ne \emptyset$}\\
        +\infty & \text{otherwise}
        \end{cases}
    \end{align*}

\noindent\textbf{Join node.}
Suppose that node $t$ is labeled with $\eta_{i,j}$ and has a child node $t'$. 
Let $H_{S}'$ be a boundary subgraph in $G_{t'}$ for a vertex subset $S \subseteq V(G_{t'})$ and let $f':\pow \times \pow \to \{0,1,2\}$ be a function that represents the number of components in $H_{S}'$.
Consider two distinct components $F_1$ and $F_2$ of $H_{S}'$.
Suppose that $F_1$ matches a pair $(C_1, X_1) \in \pow \times \pow$ with $i \in C_1$ and $F_2$ matches a pair $(C_2, X_2) \in \pow \times \pow$ with $j \in C_2$.
Let $S_1$ and $S_2$ be vertex covers of $F_1$ and $F_2$ that correspond to $X_1$ and $X_2$, respectively.
If $i \in X_1$, then $F_1$ and $F_2$ are merged into a single component of $G_t$ by an edge $v_1v_2$ for a vertex $v_1 \in V(F_1) \cap S_1$ with label $i$ and a vertex $v_2 \in V(F_2)$. 
A similar argument applies to the case $j \in X_2$.
Conversely, if $i \notin X_1$ and $j \notin X_2$, then $F_1$ and $F_2$ cannot be merged into a single component of $G_t$.
Accordingly, after the join operation, determining the number of components in the boundary subgraph $H_{S}$ of $G_t$ for $S$ is non-trivial.

To achieve this, we introduce an auxiliary graph $\hojoG_{f'}$.
For each pair $(C, X)\in \mathcal{P} \times \mathcal{P}$, add $f(C,X)$ distinct vertices labeled $(C,X)$ to $\hojoG_{f'}$.
Intuitively, each vertex of $\hojoG_{f'}$ corresponds to a component of $H_S$, ignoring the third and any subsequent components that match $(C,X)$.
We denote by $(C_v,X_v)$ the label associated with $v$.
Here, we add an edge connecting distinct two vertices $u$ and $v$ of $\hojoG_{f'}$ if and only if one of the following two condition~(1) and~(2) is satisfied: (1) $i\in X_u$ and $j\in C_v$; or (2) $j\in X_u$ and $i\in C_v$.
The definition of $E(\hojoG_{f'})$ is justified by the previous discussion.

Consider the resulting auxiliary graph $\hojoG_{f'}$.
We then find components of $\hojoG_{f'}$ and define a function $f: \pow \times \pow \to \{0,1,2\}$ as follows:
$f(C,X)$ is the smaller of $2$ and the number of components $R$ such that $R$ holds $\bigcup_{v\in R} C_v =C$ and $\bigcup_{v\in R} X_v =X$.
For a function $f: \pow \times \pow \to \{0,1,2\}$, let $A(f)$ be the set of functions $f'$ such that $f$ is obtained from the above procedure for $\hojoG_{f'}$.
The entry $\text{dp}_{t}$ is defined as follows:
\begin{align*}
    \mathrm{dp}_t(f)=
    \begin{cases}
        \min\{\mathrm{dp}_{t'}(f')\mid f' \in A(f)\}& \text{if  $A(f)\ne \emptyset$}\\
        +\infty &\text{otherwise.}
    \end{cases}
\end{align*}
\medskip\noindent\textbf{Running time.} 
Recall that it suffices to consider a pair $(C, X) \in \pow \times \pow$ with $X \subseteq C$.
Thus, the number of possible pairs is bounded by 
$\sum_{i=0}^{w}\binom{w}{i} \left( \sum_{j=0}^i\binom{i}{j} \right) 
=\sum_{i=0}^{w}\binom{w}{i}2^i=3^w$.
The function $f$ independently assigns one of three values $\{0,1,2\}$ to each pair $(C, X)$, and hence the number of possible functions is $3^{3^w}$.
Therefore, the size of the DP table is $O(3^{3^w})$ for each node $t$.
The DP table can be computed in time $O(3^{3^w})$ for an introduce node, $O(3^w(3^{3^w})^2)=O(3^w9^{3^w})$ for a union node, $O(3^w3^{3^w})$ for a relabel node, and $O((2\cdot 3^w)^2\cdot3^{3^w})=O(9^w\cdot 3^{3^w})$ for a join node.
Since the $w$-expression consists of $O(n)$ nodes~\cite{Courcelletheorem}, the total running time is $9^{O(3^{w})} n=2^{2^{O(w)}}n$.

\subsection{Full discussions of \Cref{sec:interval}} \label{app:interval}

This section shows the following theorem.

\theoreminterval*

An \emph{interval ordering} $I$ of a graph $G$ with $n$ vertices is a sequence $(v_1,v_2,\ldots, v_n)$ of vertices that satisfies the following condition: for any three integers $a,b,c\in [n]$ with $a<b<c$, if $v_av_c\in E(G)$, then $v_bv_c\in E(G)$.
A graph is an \emph{interval graph} if it has an interval ordering.
An interval ordering of a given interval graph can be computed in linear time~\cite{HsuW99}.
We give \cref{code:greedy} to solve {\mcst} for interval graphs as the proof of \Cref{the:interval}.
Note that $G$ is assumed to be connected and have at least two vertices.


\setcounter{algorithm}{0}
\begin{algorithm}[th]
	\caption{Solving \mcst\ on interval graphs}
	\label{app:code:greedy}
	\begin{algorithmic}[1]
    \Input An interval ordering $I = (v_1,v_2,\ldots, v_n)$ of a connected interval graph $G$, where $n \ge 2$
    \Output A minimum vertex cover $S$ among spanning trees of $G$
	\State $S\leftarrow\emptyset$, $V_T \leftarrow \{v_1\}$ \label{line:S}
		\While{$V_T\ne V(G)$} \label{line:while}
            \State $t_1\leftarrow\min\{i\mid v_i\notin V_T\}$\label{line:t_1}
            \State $t_2\leftarrow \max\{i\mid v_i\in V_T\}$\label{line:t_2}
            \State $t\leftarrow\min\{t_1,t_2\}$\label{line:t}
			\State $s \leftarrow \max \{i \mid v_i \in N[v_t]\}$\label{line:select}
			\State $S\leftarrow S \cup \{v_s\}$ \label{line:set_S}
			\State $V_T \leftarrow V_T \cup N[v_s]$\label{line:V_Trenew}
		\EndWhile \label{line:endwhile}
		\State \Return $S$\label{line:return}
	\end{algorithmic}
\end{algorithm}

Let $t_1^i$, $t_2^i$ $t^i$, $s^i$, $S^i$, and $V_T^i$ denote the integers $t_1$, $t_2$, $t$, $s$, and the sets $S$, $V_T$, respectively, when the $i$-th \textbf{while} loop finishes (\cref{line:endwhile}).
For the sake of convenience, we define $S^0 = \emptyset$ and $V_T^0 = \{ v_1 \}$ (that is, $S$ and $V_T$ before the first \textbf{while} loop).
We also denote $\Vint{a}{b} = \{v_a, \dots, v_b\}$ for two integers $a$ and $b$ with $1\le a \le b \le n$.

To show the correctness of \Cref{app:code:greedy}, we again utilize \cref{lem:cut}.
For a connected graph $G$ with at least two vertices, a pair $(C,S)$ of two vertex sets of $G$ is \emph{compatible} if there exist two vertices $u$ and $v$ that satisfy the following four conditions: $\mathrm{(1)}$ $uv\in E$; $\mathrm{(2)}$ $u\in C$; $\mathrm{(3)}$ $v\in V\setminus C$; and $\mathrm{(4)}$ $u\in S$ or $v\in S$.
\cref{lem:cut} can be restated as that $G$ has a spanning tree $T$ covered by a vertex subset $S$ if and only if a pair $(C,S)$ is compatible for any non-empty vertex subset $C\subset V$.

We begin by proving the following useful lemmas.

\begin{lemma} \label{lem:int_t}
    Consider when the $i$-th \textbf{while} loop in \Cref{app:code:greedy} finishes.
    It holds that $\Vint{1}{t_1^i-1} \subseteq V_T^{i-1} \subseteq \Vint{1}{t_2^i}$.
    Moreover, $\Vint{1}{s^i} \subseteq V_T^i$.
\end{lemma}

\begin{proof}
    The former statement follows from the definitions of $t_1^i$ and $t_2^i$ in \Cref{line:t_1,line:t_2}.
    
    The latter statement is proved by induction on $i$.
    In the base case $i=1$, recall that $s_1$ is adjacent to $v_1$.
    Thus, due to the interval ordering $I$, we have $\Vint{1}{s^1} \subseteq N[s^1] = V_T^1$. 

    Suppose next that $i > 1$ and $\Vint{1}{s^{i-1}} \subseteq V_T^{i-1}$.
    We show that $v_{\ell} \in V_T^{i}$ for any integer $\ell$ with $s^{i-1} <  \ell \le s^{i}$.
    For the sake of contradiction, assume that there exists an integer $\ell$ such that $s^{i-1} <  \ell \le v_{s^{i}}$ and $v_{\ell}\notin V_T^{i}$.
    Since $V_T^{i} = V_T^{i-1} \cup N[s^i]$, we also have $v_\ell \notin V_T^{i-1}$.
    By $\Vint{1}{t_1^i-1} \subseteq V_T^{i-1}$ in the former statement, we have $t_1^i \le \ell$.
    Since $t^i\le t_1^i$ from \cref{line:t}, it holds that $t^i\le t_1^i \le \ell \le s^{i}$.
    If $\ell = s^{i}$, then we have $v_\ell \in V_T^i = V_T^{i-1} \cup N[v_{s^i}]$ from \cref{line:V_Trenew}, a contradiction.
    If $\ell < s^{i}$, then $v_\ell v_{s^{i}} \in E(G)$ follows from $v_{t^i} v_{s^{i}} \in E(G)$ and the interval ordering $I$. 
    We have $v_\ell \in N[v_s^i]$, and therefore, $v_\ell \in V_T^i$, a contradiction.
\end{proof}

\begin{lemma}\label{lem:ascending}
    Suppose that \Cref{app:code:greedy} terminates after the $j$-th \textbf{while} loop in \Cref{app:code:greedy}.
    For every $i \in [j-1]$, we have $s^i < t^{i+1} <  s^{i+1}$.
\end{lemma}

\begin{proof}
    \cref{lem:int_t} leads to $\Vint{1}{t_1^{i+1}-1} \subseteq V_T^{i} \subseteq \Vint{1}{t_2^{i+1}}$ and $\Vint{1}{s^i} \subseteq V_T^i$.
    Combined with $v_{t_1^{i+1}} \notin V_T^{i}$, we obtain $s^i \le t_1^{i+1}-1 \le t_2^{i+1}$, and equivalently $s^i < t_1^{i+1} \le t_2^{i+1} + 1$.
    The case $t^{i+1} = t_1^{i+1}$ immediately yields $s^i < t^{i+1}$.
    In the case $t^{i+1} = t_2^{i+1}$, we have $s^i \le t^{i+1}$.
    If $s^i = t^{i+1}$, we also have $V_T^{i} = \Vint{1}{s^i}$.
    Considering \cref{line:V_Trenew}, the vertex $v_{s^i}$ is adjacent to none of the vertices in $V \setminus V_T^{i}$.
    This implies that all vertices in $V_T^{i}$ are adjacent to none of the vertices in $V \setminus V_T^{i}$ due to the interval ordering $I$.
    However, this contradicts that $G$ is connected.
    Thus, we conclude that $s^i < t^{i+1}$.
    The inequality $t^{i+1} < s^{i+1}$ holds from \cref{line:select}.
\end{proof}

\begin{lemma} \label{lem:int_ordering_ts}
    Consider when the $i$-th \textbf{while} loop in \Cref{app:code:greedy} finishes.
    It holds that $t^i \le t_1^i \le s^i$ and $t^i \le t_2^i \le s^i$.
    Consequently, $v_{t_1^i}, v_{t_2^i} \in N[v_{s^i}]$.
\end{lemma}

\begin{proof}
    It follows from \cref{line:t} that $t^i \le t_1^i$ and $t^i \le t_2^i$.
    For the sake of contradiction, assume that $s^i < t_1^i$.
    \cref{lem:int_t} leads to $\Vint{1}{s^i} \subseteq \Vint{1}{t_1^i-1} \subseteq V_T^{i-1} \subseteq \Vint{1}{t_2^i}$, which implies that $s^i \le t_2^i$. 
    Combined with $s^i < t_1^i$, we have $s^i \le \min \{t_1^i, t_2^i\} = t^i$.
    This contradicts $t^{i} <  s^{i}$ from \cref{lem:ascending}.
    

    Next, for the sake of contradiction, assume that $s^i < t_2^i$.
    Recall that we have $s^1 < \dots < s^{i-1} < t^i$ from \cref{lem:ascending}.
    Then, it holds that $s^1 < \dots < s^{i-1} < t^i = t_1^i \le s^i < t_2^i$.
    Since $v_{t_2^i} \in V_T^{i-1}$, some vertex in $\{s^1, \dots, s^{i-1} \}$ is adjacent to $v_{t_2^i}$.
    However, this yields $v_{t^i} v_{t_2^i} \in E(G)$, which contradicts \cref{line:select}, the maximality of $s^i$.

    In conclusion, we have $t^i \le t_1^i \le s^i$ and $t^i \le t_2^i \le s^i$.
    It follows from the interval ordering $I$ and $v_{t^i} v_{s^i} \in E(G)$ that $v_{t_1^i}, v_{t_2^i} \in N[v_{s^i}]$.
\end{proof}

We are now prepared to prove the correctness of \Cref{code:greedy}.
We first state that \Cref{code:greedy} outputs a feasible solution.

\begin{lemma}\label{lem:invariant}
    The output $S$ of \Cref{code:greedy} is a vertex cover of some spanning tree of an input interval graph $G$.
\end{lemma}
\begin{proof}
    Suppose that \Cref{code:greedy} terminates after the $j$-th \textbf{while} loop.
    We claim that a pair $(C, S^i)$ is compatible for each $i \in [j]$ and each subset $C \subseteq V(G)$ such that $V_T^i \cap C \neq \emptyset$ and $V_T^i \setminus C \neq \emptyset$ (intuitively, $C$ partially intersects with $V_T^i$).
    Since $S = S^j$ and $V(G) = V_T^j$, combined with \cref{lem:cut}, this claim leads to \cref{lem:invariant}.

    We prove the above claim by induction on $i$.
    In the base case $i = 1$, we have $S^1 = \{v_{s^1}\}$ and $V_T^1 = N[v_{s^1}]$.
    Thus, it is obvious that a pair $(C, S^1)$ is compatible for each subset $C \subseteq V(G)$ such that $N[v_{s^1}] \cap C \neq \emptyset$ and $N[v_{s^1}] \setminus C \neq \emptyset$.

    Consider the case $i > 1$ and a subset $C \subseteq V(G)$ such that $V_T^{i} \cap C \neq \emptyset$ and $V_T^{i} \setminus C \neq \emptyset$.
    Recall that $S^i = S^{i-1} \cup \{v_{s^i}\}$ and $V_T^i = V_T^{i-1} \cup N[v_{s^i}]$ from \cref{line:set_S,line:V_Trenew}.
    If $N[v_{s^i}] \cap C \neq \emptyset$ and $N[v_{s^i}] \setminus C \neq \emptyset$, then it is straightforward that a pair $(C, S^i)$ is compatible.
    Furthermore, if $V_T^{i-1} \cap C \neq \emptyset$ and $V_T^{i-1} \setminus C \neq \emptyset$, then a pair $(C, S^{i-1})$ is compatible by the induction hypothesis, and $(C, S^{i})$ is compatible as well.
    Thus, it remains to consider the following two cases: (i) $V_T^{i-1} \subseteq C$ and $N[v_{s^i}] \cap C = \emptyset$; and (ii) $ N[v_{s^i}] \subseteq C$ and $V_T^{i-1} \cap C = \emptyset$.
    In fact, cases (i) and (ii) are symmetrical, and hence we only discuss case (i).
    From the definition of $V_T^{i-1}$, we have $v_{t_2^i} \in V_T^{i-1} \subseteq C$.
    It is clear that $v_{s^i} \in N[v_{s^i}] \subseteq V(G)\setminus C$.
    Moreover, $v_{t_2^i} v_{s^i} \in E(G)$ follows from \cref{lem:int_ordering_ts}.
    Since $v_{s^i} \in S^i$, we conclude that a pair $(C, S^i)$ is compatible.
\end{proof}

Finally, we show the optimality of $S$. 

\begin{lemma}\label{app:lem:interval}
    \cref{code:greedy} outputs a minimum vertex cover among spanning trees in $G$.
\end{lemma}

\begin{proof}
    For the sake of contradiction, assume that there exists a vertex subset $S'$ of $G$ such that $|S'| < |S|$ and $G$ has a spanning tree $T$ covered by $S'$.
    
    Consider the symmetric difference $S \triangle  S' = (S \cup S')\setminus (S \cap S')$.
    Note that $S \triangle  S' \neq \emptyset$ follows from $|S'| < |S|$.
    Recall that the vertices of $G$ are in accordance with an interval ordering $I = (v_1,v_2,\ldots, v_n)$.
    Let $i$ be the smallest index of vertices in $S \triangle S'$, that is, $i =\min\{i' \mid v_{i'} \in S \triangle S'\}$.
    Then, the vertex $v_{i}$ satisfies one of the following two conditions:
    \begin{inparaenum}[(i)]
        \item $v_{i} \in S'$ and $v_{i} \notin S$; or
        \item $v_{i} \notin S'$ and $v_{i} \in S$.
    \end{inparaenum}

    Consider case (i).
    Let $s^j \in [n]$ be the minimum integer such that $v_{s^j} \in S \setminus S'$, assuming that $s^j$ is an integer $s$ selected in \cref{line:select} of the $j$-th \textbf{while} loop.
    Now, we give the following claim, whose proof is deferred to \cref{sec:clm:alttree}.

    \begin{clm}\label{clm:alttree}
        The graph $G$ has a spanning tree $T^{\ast}$ covered by $S^{\ast}=(S'\setminus \{v_{i}\})\cup \{v_{s^{j}}\}$.
    \end{clm}
    Observe that $|S^\ast| \le |S'|$ and $|S \triangle S^{\ast}| = |S\triangle S'|-1$.
    We iteratively apply the procedure above to $S'$ as long as the vertex $v_i$ satisfies case (i).
    If only case (i) appears in the procedure, we eventually obtain a vertex subset $S^{\prime\prime}$ such that $|S^{\prime\prime}| \le |S'| < |S|$ and $|S \triangle S^{\prime\prime}| = 0$.
    However, $|S \triangle S^{\prime\prime}| = 0$ means $S = S^{\prime\prime}$, a contradiction.
    
    As discussed above, we may assume that $v_i$ satisfies case (ii).
    Suppose that $i = s^j$, that is, $i$ is an integer $s$ selected in \cref{line:select} of the $j$-th \textbf{while} loop.
    In this case, $\Vint{1}{s^j-1} \cap S = \Vint{1}{s^j-1} \cap S'$ holds due to the minimality of $i = s^j$.
    We now show that $S'$ cannot be a vertex cover of any spanning tree of $G$.
    To this end, it suffices to claim that there is a vertex subset $C$ such that a pair $(C, S')$ is not compatible.
    Consider two subcases: case (ii)-(1) $t^j = t_1^j$; and case (ii)-(2) $t^j = t_2^j$.
    
    In case (ii)-(1) $t^j = t_1^j$, let $C = \{ v_{t_1^j} \}$.
    From \cref{lem:ascending}, we have $s^1 < \dots < s^{j-1} < t^j < s^j$.
    It follows from $v_{t_1^j}\notin V_T^j$ that $N[v_{t_1^j}] \cap \Vint{1}{s^j-1} \cap S = \emptyset$; otherwise, $V_T^j$ would contain $v_{t_1^j}$ from \cref{line:V_Trenew}.
    Recall that $\Vint{1}{s^j-1} \cap S = \Vint{1}{s^j-1} \cap S'$.
    Thus, we conclude that no vertex in $\Vint{1}{s^{j}-1} \cap S'$ is adjacent to $v_{t_1^j}$.
    Furthermore, since $s^j$ is the largest index among the neighbors of $v_{t_1^j}$, there is no edge between $v_{t_1^j}$ and $\Vint{s^j}{n} \cap S'$.
    Therefore, $(\{v_{t_1^j}\}, S')$ is not compatible.

    We proceed to case (ii)-(2) $t^j=t_2^j$.
    We show that $(V_T^{j-1}, S')$ is not compatible.
    For the sake of contradiction, assume that there is an edge $xy\in E$ with $x\in V_T^{j-1}$ and $y\in V\setminus V_T^{j-1}$ satisfying that $x \in S'$ or $y \in S'$.

    Suppose that $x \in S'$. 
    Recall that $S \cap \Vint{1}{s^j-1} = S' \cap \Vint{1}{s^j-1}$ holds.
    Since $S \cap \Vint{1}{s^j-1} = S^{j-1}$ and $v_{s^j} \notin S'$, this implies that $S^{j-1} = S' \cap \Vint{1}{s^j}$.
    Here, it can be seen that $V_T^{j-1} \subseteq \Vint{1}{t_2^j} \subseteq \Vint{1}{s^j}$ from \cref{lem:int_t,lem:int_ordering_ts}.
    Thus, we have $S' \cap V_T^{j-1} \subseteq S' \cap \Vint{1}{s^j} = S^{j-1}$.
    Since $x \in S' \cap V_T^{j-1}$, we conclude that $x \in S^{j-1}$.
    However, this yields that $y \in V_T^{j-1}$ from \ref{line:V_Trenew}, a contradiction.
    
    Suppose next that $x\notin S'$ and $y\in S'$.
    Let $i_x$ and $i_y$ be the indices of $x$ and $y$, respectively.
    From \cref{lem:int_t} and $x \in V_T^{j-1}$, we have $i_x \le t_2^{j}$.
    Moreover, since $y \notin V_T^{j-1}$, we have $t_1^j \le i_y$.
    Thus, $i_x \le t_2^j < t_1^j \le i_y$ follows from $t^j = t_2^j < t_1^j$, and hence $xy \in E(G)$ implies $v_{t_2^j}y \in E(G)$ due to the interval ordering $I$.
    The maximality of $s^j$ implies that $i_y \le s^j$.
    In addition, we have $s^{j-1} < t^j$ from \cref{lem:ascending}.
    Consequently, the vertex $y$ satisfies $y \in S'$ and $s^{j-1} < t^j < i_y < s^j$.
    However, $s^{j-1} < i_y < s^j$ leads to $y \notin S$, and therefore $y \notin S'$ is derived from $\Vint{1}{s^j-1} \cap S = \Vint{1}{s^j-1} \cap S'$, a contradiction. 
\end{proof}

\subsubsection{Proof of \Cref{clm:alttree}} \label{sec:clm:alttree}
\begin{proof}
We show that a pair $(C, S^\ast)$ is compatible for any non-empty subset $C \subset V(G)$.
We first consider the case $C \cap V_T^j \neq \emptyset$ and $C \setminus V_T^j \neq \emptyset$.
Then a pair $(C, S)$ is compatible for any non-empty subset $C \subset V(G)$.
This follows from the fact that each vertex in $V_T^j$ is adjacent to some vertex in $S^j \subseteq V_T^j$. 
We also note that $\Vint{1}{s^j} \cap S = \Vint{1}{s^j} \cap S^\ast$ from the definitions of $i$, $s^j$, and $S^\ast$.
This implies that $S^j \subseteq S^\ast$, and consequently, a pair $(C, S^\ast)$ is also compatible.

Suppose otherwise, that is, $V_T^j \subseteq C$ or $V_T^j \subseteq V\setminus C$.
By symmetry, it suffices to consider the case $V_T^j \subseteq C$.
In this case, $v_{s^j} \in V_T^j \subseteq C$ follows from $\Vint{1}{s^j} \subseteq V_T^j$ in \cref{lem:int_t}.
Moreover, since $i < s^j$, we have $v_i \in V_T^j \subseteq C$.
We may also assume that $v_{t_1^j}, v_{t_2^j} \in C$; if one of $v_{t_1^j}$ and $v_{t_2^j}$ belongs to $V \setminus C$, then a pair $(C, S^\ast)$ is compatible due to $v_{s^j} v_{t_1^j} \in E(G)$ or $v_{s^j} v_{t_2^j} \in E(G)$ from \cref{lem:int_ordering_ts}.

Here, consider a pair $(C, S')$.
From the assumption of $S'$, the pair $(C, S')$ is compatible.
More precisely, there exist two vertices $x$ and $y$ such that
$\mathrm{(1)}$ $xy\in E(G)$; $\mathrm{(2)}$ $x\in C$; $\mathrm{(3)}$ $y\in V\setminus C$; and $\mathrm{(4)}$ $x\in S'$ or $y\in S'$. 
Recall that $v_i \in  C$.
If $x \neq v_i$, then $(C, S^\ast)$ is also compatible because $x\in S'$ implies $x \in S^\ast$.
Suppose that $x = v_i$.
Here, note that $y \in V\setminus C$ and hence $y \notin V_T^j$.
Combined with $\Vint{1}{s^j} \subseteq V_T^j$ in \cref{lem:int_t}, 
we have $s^j < i_y$, where $i_y$ is the index of $y$ according to the interval ordering $I$.
Since $i < s^j$, we have $i < s^j < i_y$.
From the fact that $xy \in E(G)$ and $x = v_i$, we also have $v_{s^j} y \in E(G)$, where $v_{s^j} \in S^\ast$.
Therefore, a pair $(C,S^\ast)$ is compatible.
\end{proof}

\end{document}